\documentclass[final,runningheads]{llncs}

\usepackage{ifthen}
\usepackage{latexsym}
\usepackage{orcidlink}
\usepackage[misc,geometry]{ifsym}
\newboolean{showappendix}

\setboolean{showappendix}{true}

\title{An $\Flsharp$ Based Algorithm for Active Learning of Minimal Separating Automata}
\titlerunning{Active Learning of Minimal Separating Automata}

\author{Jasper Laumen \and Leonne Snel \and Frits Vaandrager$^{\textrm{(\Letter) \orcidlink{0000-0003-3955-1910}}}$   }
\authorrunning{Jasper Laumen \and Leonne Snel \and Frits Vaandrager}
\institute{
	Institute for Computing and Information Sciences \\
	Radboud University, Nijmegen, The Netherlands\\ \email{\{Jasper.Laumen@ru.nl,Leonne.Snel@ru.nl,Frits.Vaandrager@ru.nl\}}} 

\usepackage{booktabs}
\usepackage{mathtools}
\usepackage[outline]{contour}
\contourlength{1.2pt}
\usepackage{tikz}
\usetikzlibrary{patterns, intersections, positioning, shapes.geometric, calc,
	automata, arrows, decorations.pathreplacing,decorations.pathmorphing}
\usetikzlibrary{cd}

\newcommand{\treeNodeLabel}[1]{\contour{white}{#1}}

\tikzset{
	initial text={},
	treenode/.style = {align=center, inner sep=0pt, text centered},
	basis/.style = {
		pattern=north east lines,
		pattern color=magenta!80!black!80!white,
	},
	frontier/.style={
		pattern=crosshatch dots,
		pattern color=yellow!80!black,
	},
	q0class/.style={
		pattern=vertical lines,
		pattern color=red!60!white,
	},
	q1class/.style={
		pattern=north west lines,
		pattern color=blue!60!white,
	},
	q2class/.style={
		pattern=grid,
		pattern color=green!50!black!60!white,
	},
	q3class/.style={
		pattern=horizontal lines,
		pattern color=yellow!50!black!60!white,
	},
	q4class/.style={
		pattern=crosshatch dots,
		pattern color=red!80!black,
	},
}

\usepackage{hyperref}
\usepackage[capitalise]{cleveref}
\usepackage{xspace}
\newcommand{\lsharp}{$L^\#$\xspace}
\newcommand{\lsharps}{$L^\#_\square$\xspace}
\newcommand{\Flsharp}{L^{\#}}
\newcommand{\lsharpb}{L^{\#}_{\Box}}
\newcommand{\lstar}{$L^*$\xspace}
\newcommand{\lstars}{$L^*_\square$\xspace}

\newcommand{\tree}{\mathcal{T}}

\newcommand{\darrow}{\mathord{\downarrow}}

\newcommand{\apart}{\mathbin{\#}}
\newcommand{\access}{\textsf{access}}

\newcommand{\Nat}{\mathbb{N}}
\newcommand{\statemerge}[3]{\mathsf{merge}(#1, #2, #3)}

\usepackage{algorithm}
\usepackage{algpseudocode}

\RequirePackage[rgb]{xcolor}
\definecolor{redimpact}{RGB/CMYK}{227,0,11/0,100,100,0}
\definecolor{poppy}{RGB/CMYK}{255,66,75/0,80,60,0}
\definecolor{ladybug}{RGB/CMYK}{190,49,26/0,92,90,27}
\definecolor{berry}{RGB/CMYK}{143,32,17/15,92,90,40}
\definecolor{maroon}{RGB/CMYK}{115,14,4/20,100,80,50}
\definecolor{mahogany}{RGB/CMYK}{74,0,4/45,100,75,70}

\RequirePackage{amsmath}
\RequirePackage{amssymb}
\newcommand{\set}[1]{\left\{#1\right\}}
\newcommand{\br}[1]{\left(#1\right)}
\newcommand{\abs}[1]{\left|#1\right|}
\newcommand{\A}{\mathcal{A}}
\newcommand{\B}{\mathcal{B}}
\newcommand{\C}{\mathcal{C}}
\newcommand{\Hyp}{\mathcal{H}}

\newcommand{\bigO}{\mathcal{O}}
\newcommand{\Obs}{\mathcal{T}}

\newcommand{\defined}{\ensuremath{\mathord{\downarrow}}}
\newcommand{\converges}{\ensuremath{\mathord{\downarrow}}}
\newcommand{\diverges}{\ensuremath{\mathord{\uparrow}}}
\newcommand{\partialto}{\ensuremath{\rightharpoonup}}

\begin{document}

\maketitle

\begin{abstract}
  A DFA separates two disjoint languages $L_1$ and $L_2$ if it accepts every word in $L_1$ and rejects every word in $L_2$. Algorithms for active learning of small separating DFAs have many applications, e.g.\  for learning network invariants, learning contextual assumptions in compositional verification, learning state machines from large amounts of log data, and learning bug pattern descriptions.  We propose a simple active learning algorithm, inspired by $L^{\#}$, that learns a minimal separating DFA for disjoint languages $L_1$ and $L_2$ if one exists. Experiments show that our algorithm significantly outperforms existing active learning algorithms on both randomly generated and industrial benchmarks. 
\end{abstract}

\section{Introduction}
\label{sec:introduction}
Let $L_1$ and $L_2$ be disjoint languages over some alphabet $\Sigma$.  We say that a deterministic finite automaton (DFA) $\A$ \emph{separates} $L_1$ and $L_2$ if it accepts every word in $L_1$ and rejects every word in $L_2$, that is, $L_1 \subseteq L(\A)$ and $L_2 \cap L(\A) = \emptyset$, with $L(\A)$ the language accepted by $\A$.
A minimal separating DFA is a separating DFA that has the least number
of states among all separating DFAs.
Algorithms for finding a minimal (or small) separating DFA have many applications, e.g.\ for computing network invariants \cite{GrinchteinLP06},  contextual assumptions in compositional verification \cite{ChenFCTW09}, minimal DFAs that are consistent with a set of positively and negatively classified samples \cite{Hig10}, and bug pattern descriptions \cite{YaacovWAH25}.
Finding a minimal separating DFA is nontrivial: for regular languages $L_1$ and $L_2$, the decision problem ``Does a separating DFA with $k$ states exist?'' is NP-hard \cite{Pfleeger73,nphard}. If languages $L_1$ and $L_2$ are context-free, then the decision problem whether there exists a separating DFA is already undecidable \cite{LeuckerN12}.

SAT and SMT solvers may be used to compute minimal separating DFAs if $L_1$ and $L_2$ are regular \cite{Neider12}.  
However, in applications the minimal DFAs for $L_1$ and $L_2$ are typically big, whereas a minimal separating DFA $\A$ is small.
Since solvers are unable to handle such large instances of an NP-hard problem, researchers have resorted to the use of active learning algorithms for the construction of a separating DFA. Different types of queries have been used, also depending on the characteristics of the applications.
A natural setting for learning separating DFAs is obtained by considering a slight variation of the Minimally Adequate Teacher (MAT) framework of \cite{Ang87} in which a \emph{learner} tries to identify a minimal separating DFA  by posing two types of queries to a \emph{teacher} \cite{LeuckerN12,lstars}:
\begin{itemize}
	\item 
	A \emph{membership query} $w \in \Sigma^*$.
	In response, the teacher answers ``$+$'' if $w \in L_1$, ``$-$'' if $w \in L_2$, and ``don't care'' (or ``don't now'', ``maybe'', etc), denoted by $\square$, in any other case.
	\item 
	A \emph{validity query}  for a hypothesis DFA $\Hyp$.
	In response, the teacher answers ``yes'' if $\Hyp$ separates $L_1$ and $L_2$, and otherwise returns a counterexample, that is, a sequence $w$ such that $w \in L_1 \setminus L(\Hyp)$ or $w \in L_2 \cap L(\Hyp)$.
\end{itemize}
A first active learning algorithm for learning a minimal separating DFA for regular languages $L_1$ and $L_2$ using membership and validity queries was proposed by Grinchtein \emph{et al.} \cite{GrinchteinLP06}.
Their algorithm, which is inspired by earlier work of Pena \& Oliveira \cite{PenaO99}, adapts Angluin's $L^*$ algorithm \cite{Ang87} by extending the observation table with don't care values ($\square$) besides accept ($+$) and reject values ($-$), and uses the algorithm of Biermann \& Feldman \cite{BiermannF72} to construct a DFA that is consistent with the observation table.
Chen \emph{et al.} \cite{ChenFCTW09} presented an $L^*$-based algorithm using 3DFAs that significantly outperformed the algorithm of Grinchtein \emph{et al.} \cite{GrinchteinLP06}.
Their algorithm also uses membership queries, but instead of validity queries it poses the (strictly more powerful) \emph{containment queries}:
\begin{itemize}
\item 
A \emph{containment query} for a hypothesis DFA $\Hyp$ has one of the following four types:
(i) $L_1 \subseteq L(\Hyp)$, (ii) $L(\Hyp) \subseteq L_1$, (iii) $\overline{L_2} \subseteq L(\Hyp)$, and (iv) $L(\Hyp) \subseteq \overline{L_2}$. The teacher returns ``yes'' if the containment holds, and provides a counterexample otherwise.
\end{itemize} 
Leucker \& Neider \cite{LeuckerN12} presented a survey of learning algorithms for separating DFAs, suggested some improvements of the algorithms of \cite{GrinchteinLP06,ChenFCTW09}, but did not explore any implementation or benchmarking.
Recently, however, Yaacov \emph{et al.} \cite{YaacovWAH25} implemented an improvement of the algorithm of \cite{ChenFCTW09} suggested by Leucker \& Neider \cite{LeuckerN12}, and used it for producing compact and informative bug descriptions for a collection of industrial benchmarks taken from \cite{rers}.
However, in order to make their approach sufficiently scalable, they used the state merging algorithm RPNI \cite{rpni} (rather than SAT/SMT solving as suggested in \cite{LeuckerN12}) to obtain small (rather than minimal) separating DFAs.
Moeller \emph{et al.} \cite{lstars} adapted Angluin's $L^*$ algorithm for learning minimal separating DFAs using membership and validity queries, and evaluated their \lstars algorithm on the benchmark collection of Oliveira \& Silva \cite{benchmarks} and some smaller benchmarks of Lee \emph{et al.} \cite{LeeSO16}.  Their implementation can only handle about half of the benchmarks of \cite{benchmarks}, due to scalability issues.
An interesting new type of queries was recently proposed by Walinga \emph{et al.} \cite{abs-2406-07208}. They consider the problem of learning a state machine from an amount of log data that is so big that conventional state merging algorithms cannot handle it. In order to solve this problem, they use database technology to efficiently query a large dataset stored on disk, and
pose \emph{database queries} of the following form:
\begin{itemize}
	\item 
	A \emph{database query} is a triple $(w, n, k)$ with $w \in \Sigma^*$ and $n, k \in \Nat$. The teacher returns at most $k$ (classified) words of length at most $n$ with $w$ as a prefix.
\end{itemize}
Note that database queries are incomparable to membership, validity and containment queries.
The DAALder algorithm of \cite{abs-2406-07208} uses database queries to grow  a prefix-tree acceptor (PTA) in a manner inspired by the \lsharp algorithm \cite{VaandragerGRW22}. The state merging tool FlexFringe \cite{VerwerH17} is then used to turn this PTA into a small state machine. DAALder is evaluated on benchmarks with artificial trace data.

The main result of this article is $\lsharpb$, a new and simple algorithm inspired by \lsharp \cite{VaandragerGRW22,lsharp}, for learning a minimal separating DFA 
using membership and validity queries.
Our algorithm is based on the following key ideas:
\begin{enumerate}
\item 
By storing query results in an observation tree (PTA) and using apartness of basis states as the main driver for membership queries, our algorithm has more options to make progress with learning than $L^*$-based algorithms, which are limited by the prefixes and suffices in an observation table; this helps when the teacher gives many ``don't care'' answers.
\item
We don't need the assumption of \lsharp that the basis is prefix closed.
\item
We don't need a worklist of candidate trees/tables as used in \lstars  \cite{lstars}.
\item 
Both the $\lsharpb$ algorithm and its correctness proof can be elegantly phrased in terms of morphisms between 3DFAs.
\item
By adding redundant clauses, which assert that states can only be merged when they are not apart, we significantly improve the SMT encoding of \cite{smt} of the problem whether a morphism exists from a PTA to a DFA with $n$ states.
\end{enumerate}
We also present two further optimisations of $\lsharpb$: a first one in which basis states are required to be pairwise incompatible rather than pairwise apart, and a second one in which we replace basis states by states that are closer to the root of the observation tree.

Experiments show that $\lsharpb$ outperforms 
\lstars  \cite{lstars} on the benchmarks of \cite{benchmarks}, and the algorithm of \cite{YaacovWAH25} on the benchmarks of \cite{rers}.
In particular, unlike \lstars, our $\lsharpb$ algorithm can solve almost all the benchmarks (within 200 seconds) of \cite{benchmarks}, and the running times of $\lsharpb$ for the RERS benchmarks \cite{rers} are up to two orders of magnitude faster than for the algorithm of \cite{YaacovWAH25}, even though our separating DFAs are minimal, unlike those  of \cite{YaacovWAH25}.
So our results really move the boundary of applications for which minimal separating DFAs can be computed.

The rest of this paper is organized as follows. After some preliminary definitions in \cref{sec:preliminaries}, we introduce the $\lsharpb$ algorithm in \cref{sec:algorithm}, and prove its correctness.
\cref{optimisations} presents the two optimisations of $\lsharpb$.
In \cref{evaluation} we evaluate the performance of $\lsharpb$ on two groups of benchmarks, and compare it with the performance of algorithms used in \cite{lstars,YaacovWAH25}.
We also perform an ablation-style study to evaluate the effect of the optimisations of \cref{optimisations} and the use of redundant clauses by the SMT solver.
Finally, in \cref{sec:discussion}, we discuss our results and present some directions for future research.
\ifshowappendix
Appendix~\ref{sec:incompatibility}
\else
The full version of this article \cite{LSV26report}
\fi
provides additional details on the first optimisation of \cref{optimisations}, and the connections between incompatibility and apartness.

\section{Preliminaries}
\label{sec:preliminaries}
This section first fixes notation for partial maps and sequences, and then introduces some basic automata theory concepts that we will need in this paper.

We write $f \colon X \partialto Y$ to denote that $f$ is a partial function from $X$
to $Y$ and write $f(x) \converges$ to mean that $f$ is defined on $x$, that is,
$\exists y \in Y \colon f(x)=y$, and conversely write $f(x)\diverges$ if $f$ is
undefined for $x$.
Function $f \colon X \partialto Y$ is \emph{total}, and we write $f \colon X \to Y$, if $f(x) \converges$, for all $x \in X$.
Often, we identify a partial function $f \colon X \partialto Y$ with the set $\{ (x,y) \in X \times Y \mid f(x)=y \}$. The composition of partial
maps $f\colon X\partialto Y$ and $g\colon Y\partialto Z$ is denoted by $g\circ
f\colon X\partialto Z$, and we have $(g\circ f)(x)\converges$ iff $f(x)\converges$
and $g(f(x)) \converges$.
We use the Kleene equality on partial functions, which states that on a given argument either both functions are undefined, or both are defined and their values on that argument are equal. 

Throughout this article, we fix a finite input alphabet $\Sigma$.
We use standard notations for sequences.  If $X$ is a set then $X^*$ denotes the set of finite \emph{sequences} (also called \emph{words}) over $X$. 
We write $\epsilon$ to denote the empty sequence, 
$x$ to denote the sequence consisting of a single element $x \in X$, and $\sigma \cdot \rho$ (or simply $\sigma  \rho$) to denote the concatenation of sequences $\sigma, \rho \in X^*$. 

Below we introduce \emph{3-valued nondeterministic finite automata (3NFA)}, a generalization of the 3DFAs of \cite{ChenFCTW09} to the nondeterministic setting.  Although this paper mostly considers deterministic automata, nondeterminism pops up in \Cref{optimisations} when we consider state merging.
A 3NFA is basically a finite automaton whose states can be either accepting, rejecting, or have an ``unknown'' status.

\begin{definition}[3NFA]
	\label{3NFA}
	A \emph{3-valued nondeterministic finite automaton (3NFA)} $\A$  is a tuple $(Q, q^0, F, \rightarrow)$, where
	\begin{itemize}
		\item
		$Q$ is a finite set of \emph{states},
		\item
		$q^0 \in Q$ is the \emph{initial} state,
		\item
		$F \colon Q \partialto \{+,-\}$, is the \emph{final state function}, and 
		\item
		$\mathord{\rightarrow} \subseteq Q \times \Sigma \times Q$ is the \emph{transition relation}.
	\end{itemize}
	A state $q$ is called \emph{accepting} if $F(q)=+$ and \emph{rejecting} if $F(q)= -$. 
	States $p$ and $q$ are called \emph{conflicting} if $F(p)\converges \wedge F(q)\converges \wedge F(p) \neq F(q)$, and \emph{nonconflicting} otherwise.
	We write $q \xrightarrow{a} q'$ if $(q, a, q') \in \mathord{\rightarrow}$.
	We say that $\A$ is \emph{deterministic} iff, for all $q, q', q'' \in Q$ and $a \in \Sigma$, $q \xrightarrow{a} q'$ and $q \xrightarrow{a} q''$ implies $q'= q''$. In this case, we refer to $\A$ as a \emph{3-valued deterministic finite automaton (3DFA)}.
	An \emph{NFA} (resp.\ \emph{DFA}) is a 3NFA (resp.\ 3DFA) for which function $F$ is total.
	We say that 3NFA $\A$ is \emph{complete} iff function $F$ is total and, for each $q \in Q$ and each $a \in \Sigma$, there is a $q'\in Q$ such that $q \xrightarrow{a} q'$.\footnote{Note that we do not require DFAs and NFAs to be complete.}
	We generalize the transition relation to words of arbitrary length, and define $\mathord{\Rightarrow} \subseteq Q \times \Sigma^* \times Q$ as the smallest relation such that, for all $q, q', q'' \in Q$, $w \in \Sigma^*$ and $a \in \Sigma$,
	\begin{enumerate}
		\item
		$q \xRightarrow{\epsilon} q$
		\item 
		if $q \xrightarrow{a} q'$ and $q' \xRightarrow{w} q''$ then $q \xRightarrow{a \; w} q''$
	\end{enumerate}
	The \emph{language} of $\A$, designated $L(\A)$, is the set
	$\{ w \in \Sigma^* \mid \exists q \in Q : q^0 \xRightarrow{w} q \}$.
	A language $L \subseteq \Sigma^*$ is \emph{regular} if there is a DFA $\A$ with $L = L(\A)$.
	Subscript $\A$ is used to disambiguate to which 3NFA we refer, e.g.\  $Q_{\A}$ and $\rightarrow_{\A}$.
\end{definition}


We consider maps between 3NFAs that preserve initial states, the final state function, and the transition function.

\begin{definition}[Morphism]
	\label{def functional simulation}
	For 3NFAs $\A$ and $\B$, a \emph{morphism} (a.k.a.\ a \emph{functional simulation}) from $\A$ to $\B$ is a function $f\colon Q_{\A} \to Q_{\B}$ such that, for all $q, q'\in Q_{\A}$ and $a \in \Sigma$,
	
	\begin{enumerate}
		\item $f(q^0_{\A}) = q^0_{\B}$, 
		\item $F_{\A}(q)\defined$ $\implies$ $F_{\A}(q) = F_{\B}(f(q))$, and
		\item $q \xrightarrow{a}_{\A} q'$ $\implies$ $f(q) \xrightarrow{a}_{\B} f(q')$.
	\end{enumerate}
	We write $\A \sqsubseteq_f \B$ if $f$ is a morphism from $\A$ to $\B$, and $\A \sqsubseteq \B$ if there exists a morphism from $\A$ to $\B$.
\end{definition}

The next lemma is a straightforward consequence of the definitions.

\begin{lemma}
	\label{preorder}
	$\sqsubseteq$ is a preorder on 3NFAs, i.e., a reflexive and transitive relation.
\end{lemma}


A \emph{prefix-tree acceptor (PTA)} is a tree-shaped 3DFA. PTAs are commonly used as a datastructure to represent the results of membership queries.

\begin{definition}[PTA]
	A 3DFA $\Obs$ is a \emph{prefix tree acceptor (PTA)} iff for each $q \in Q_{\Obs}$ there is a unique sequence of inputs $w \in \Sigma^*$ with $q^0_{\Obs} \xRightarrow{w} q$.
	We write $\mathsf{access}(q)$ to denote this unique \emph{access sequence}.
	We say that $\Obs$ is a \emph{prefix tree acceptor} for a 3NFA $\A$ if $\Obs$ is a PTA and there exists a morphism $f\colon \Obs\to \A$.
\end{definition}

\begin{example}
	\cref{Fig:Obs}(bottom) shows a complete DFA $\A$ with 4 states that accepts words over alphabet $\{ a, b \}$ that contain both an even number of $a$'s and an even number of $b$'s.
	\cref{Fig:Obs}(top) gives an example of a PTA $\Obs$ for $\A$.
	Accepting states are indicated by a double circle and unknown states are indicated by a dashed circle.
	Morphism $f$ is indicated via coloring of the states. 
	\begin{figure}[h] 
		\begin{center}
			\begin{tikzpicture}[->,>=stealth',shorten >=1pt,auto,node distance=1.5cm,main node/.style={circle,draw,font=\sffamily\large\bfseries},
				]
				\def\yoffset{8mm}
				\node[initial,accepting,state,frontier] (0) {\treeNodeLabel{$t_0$}};
				\node[state, q1class] (1) [right of=0,yshift=\yoffset] {\treeNodeLabel{$t_1$}};
				\node[state, q2class] (2) [right of=0,yshift=-\yoffset] {\treeNodeLabel{$t_2$}};
				\node[state, basis] (3) [right of=2] {\treeNodeLabel{$t_3$}};
				\node[state,dashed, frontier] (5) [right of=2,yshift=2*\yoffset] {\treeNodeLabel{$t_5$}};
				\node[state, q1class] (4) [right of=3] {\treeNodeLabel{$t_4$}};
				\node[state,accepting, frontier] (6) [right of=4,yshift=2*\yoffset] {\treeNodeLabel{$t_6$}};
				\node[state, basis] (7) [right of=4] {\treeNodeLabel{$t_7$}};
				\node[state, dashed,q2class] (8) [right of=3,yshift=2*\yoffset] {\treeNodeLabel{$t_8$}};
				\node[state,q1class] (11) [right of=5, yshift=2*\yoffset] {\treeNodeLabel{$t_{11}$}};
				\node[state,q2class] (12) [right of=1, yshift=2*\yoffset] {\treeNodeLabel{$t_{12}$}};
				\node[state, accepting,frontier] (9) [right of=8, yshift=2*\yoffset] {\treeNodeLabel{$t_9$}};
				\node[state, q2class] (10) [right of=6, yshift=2*\yoffset] {\treeNodeLabel{$t_{10}$}};
				
				\node[initial,accepting,state,frontier] [below of=2, yshift=-2*\yoffset](q0) {\treeNodeLabel{$q_0$}};
				\node[state,q1class] (q1) [right of=q0,yshift=\yoffset] {\treeNodeLabel{$q_1$}};
				\node[state,q2class] (q2) [right of=q0,yshift=-\yoffset] {\treeNodeLabel{$q_2$}};
				\node[state,basis] (q3) [right of=q1,yshift=-\yoffset] {\treeNodeLabel{$q_3$}};
				
				\node[anchor=base] (f) at ($ (3.base) !.5! (q1.base) $) {\begin{tikzcd}
						{} \arrow{r}{f}
						&[8mm] {}
				\end{tikzcd}};
				
				\path[every node/.style={font=\sffamily\scriptsize}]
				(0) edge node[sloped,above] {$b$} (1)
				(0) edge node[sloped,below] {$a$} (2)
				(2) edge node[below] {$b$} (3)
				(2) edge node[sloped,above] {$a$} (5)
				(3) edge node [below] {$a$} (4)
				(3) edge node[sloped,above] {$b$} (8)
				(4) edge node[sloped,above] {$b$} (6)
				(4) edge node [below] {$a$} (7)
				(5) edge node[sloped,above] {$b$} (11)
				edge node [left] {$a$} (12)
				(8) edge node[sloped,above] {$a$} (9)
				(6) edge node[sloped,above] {$a$} (10)
				(q0) edge [bend left, above] node {$b$} (q1)
				(q0) edge [left, above] node {$a$} (q2)
				(q1) edge [bend left, above] node {$a$} (q3)
				(q1) edge [left, below] node {$b$} (q0)
				(q2) edge [bend left, below] node {$a$} (q0)
				(q2) edge [left, above] node {$b$} (q3)
				(q3) edge [left, below] node {$a$} (q1)
				(q3) edge [bend left, below] node {$b$} (q2);
			\end{tikzpicture}
			\caption{A prefix tree acceptor $\Obs$ (top) for a DFA $\A$ (bottom).}
			\label{Fig:Obs}
		\end{center}
	\end{figure}
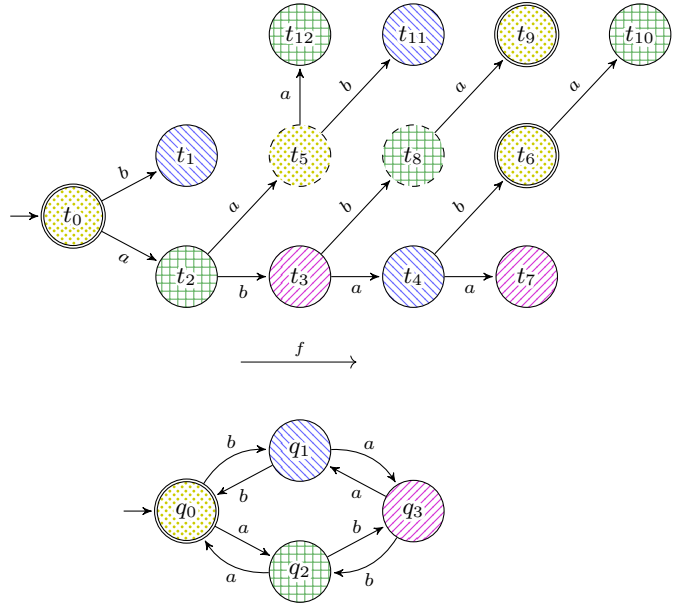
\end{example}

Two states of a 3NFA are \emph{compatible} if there exists a morphism that maps both of them to the same state of some 3DFA.

\begin{definition}[Compatibility]
	Let $\A$ be a 3NFA $\A$. We call states $p$ and $q$ of $\A$ \emph{compatible} if $\A \sqsubseteq_f \B$, for some 3DFA $\B$ and morphism $f$ with $f(p)=f(q)$.
	If $p$ and $q$ are not compatible then we say they are \emph{incompatible}.
\end{definition}

Apartness \cite{troelstra_schwichtenberg_2000,GJapartness,VaandragerGRW22} is a constructive form of incompatibility where we actually have a \emph{witness} that proves the incompatibility.

\begin{definition}[Apartness]
	For a 3NFA $\A$, we say that states $p, q \in Q_{\A}$ are  \emph{apart}
	(written $p \apart q$) if there is some $w \in \Sigma^*$ and $p', q' \in Q_{\A}$ such that $p \xRightarrow{w} p'$, $q \xRightarrow{w} q'$, and $p'$ and $q'$ are conflicting.
	We call $w$ a \emph{witness} for $p \apart q$ (and a \emph{separating sequence}) and write $w \vdash p\apart q$.
\end{definition}

\begin{example}
	For the PTA of \cref{Fig:Obs}, we may derive a number of apartness pairs and corresponding witnesses.  
	Witness $\epsilon$ demonstrates that an accepting state is always apart from a rejecting state.  Thus, for instance:
	\[
	\epsilon \vdash t_0 \apart t_2 \qquad
	\epsilon \vdash t_0 \apart t_3 \qquad
	\epsilon \vdash t_0 \apart t_4 \qquad
	\epsilon \vdash t_2 \apart t_9 \qquad
	\]
	States can be apart, irrespective of whether they are unknown, accepting, or rejecting:
	\[
	a \vdash t_5 \apart t_8 \qquad
	a \vdash t_8 \apart t_6 \qquad
	b \vdash t_5 \apart t_4 \qquad
	b \vdash t_2 \apart t_4 \qquad
	\]
	Hence by ``back propagation'':
	\[
	a\;b \vdash t_2 \apart t_3 \qquad
	b\;a \vdash t_3 \apart t_4 \qquad
	\]
	The apartness pairs listed above show that $t_0$, $t_2$, $t_3$ and $t_4$ are pairwise apart.
\end{example}

The next lemma asserts that apartness implies incompatibility.
The converse implication does not hold \cite{RW23}; we will discuss an example in \cref{sec:promotion based on incompatibility}.

\begin{lemma}
	\label{la: apartness refinement}
	Let $\A$ be a 3NFA with states $p$ and $q$.
	Then $p \apart q$ implies that $p$ and $q$ are incompatible. 
\end{lemma}
\begin{proof}
	Assume that $p \apart q$.
	Then there are $w \in \Sigma^*$ and $p', q' \in Q_{\A}$ such that $p \xRightarrow{w} p'$, $q \xRightarrow{w} q'$, and $p'$ and $q'$ are conflicting.
	We prove that $p$ and $q$ are incompatible by induction on the length of $w$:
	\begin{itemize}
		\item 
		If $w$ has length $0$ then $w = \epsilon$, $p = p'$, $q = q'$, and so $p$ and $q$ are conflicting.
		Suppose  $\A \sqsubseteq_f \B$, for some 3DFA $\B$ and morphism $f$. 
		Then \Cref{lemma identified by simulation implies mergeable} implies $f(p) \neq f(q)$.
		Hence $p$ and $q$ are incompatible.
		\item 
		Suppose $w = a \cdot w'$, for some $a \in \Sigma$ and $w'\in \Sigma^*$.
		Then there are states $p''$ and $q''$ such that
		$p \xrightarrow{a} p'' \xRightarrow{w'} p'$ and $q \xrightarrow{a} q'' \xRightarrow{w'} q'$. Note that $p'' \apart q''$.
		By the induction hypothesis, $p''$ and $q''$ are incompatible.
		Assume $p$ and $q$ are compatible.
		Then $\A \sqsubseteq_f \B$ and $f(p)=f(q)$, for some 3DFA $\B$ and morphism $f$.
		Since $f$ is a morphism, $f(p) \xrightarrow{a} f(p'')$ and $f(q) \xrightarrow{a} f(q'')$.
		Because $f(p)=f(q)$ and $\B$ is deterministic, we conclude  $f(p'') = f(q'')$.
		But, since $p''$ and $q''$ are incompatible, we also have $f(p'') \neq f(q'')$.
		Contradiction, and we conclude that $p$ and $q$ are incompatible.
	\end{itemize}
\end{proof}

\section{The \texorpdfstring{\lsharps}{L\#-box} Algorithm}
\label{sec:algorithm}
In this section we introduce $\lsharpb$, an algorithm that learns a minimal separating DFA for two disjoint languages $L_1$ and $L_2$ (if it exists) using membership and validity queries as defined in the introduction.  Our algorithm is inspired by $\Flsharp$ \cite{VaandragerGRW22,lsharp}, a recent algorithm for active learning of DFAs.
The main challenge compared to the traditional MAT framework of Angluin \cite{Ang87} for learning DFAs is that a teacher can respond with ``don't care'' to membership queries.  Existing active learning algorithms, $\Flsharp$ included, cannot handle this type of uncertainty. 

\subsection{Algorithm}
$\lsharpb$ maintains the following data structures:
\begin{enumerate}
\item 
A PTA $\Obs= (Q, q^0, F, \rightarrow)$, called the \emph{observation tree}, to store the outcomes of all membership queries: 
after a membership query for a word $w$, we add unique states to $\Obs$ to represent $w$ and all its prefixes (unless they are already present). 
Initial state $q^0$ represents the empty sequence $\epsilon$. 
Whenever state $q'$ represents sequence $w a$ and state $q$ represents $w$, then we add a transition $q \xrightarrow{a} q'$ to $\Obs$.
In this way each state $q$ represents its own access sequence.
If the response to a query $w$ is $+$ and $q$ represents $w$, then we set $F(q)=+$.
Similarly, if the response to a query $w$ is $-$ and $q$ represents $w$, then we set $F(q)=-$.
The final state function $F$ is undefined for all other states of $\Obs$,
that is, states that either represent queries for which a $\square$-response was received, or do not represent queries but rather proper prefixes of queries.
\Cref{Fig:Obs}(top) shows an example of the observation tree obtained after receiving $+$-responses for queries $\epsilon$, $abab$ and $abba$, and $-$-reponses for queries $a$, $b$, $ab$, $aaa$, $aab$, $aba$, $abaa$ and $ababa$.
Initially, in $\lsharpb$, the observation tree has just a single (initial) state $q^0$ with $F(q^0)$ undefined.
\item 
A set of states $Q_{\square}$ that correspond to queries for which a $\square$-response was received.
Once a state $q$ is in $Q_{\square}$, then the learner knows that a query  $\mathsf{access}(q)$ will provide no new information as the teacher does not care whether $w$ is accepted or not.
Initially, $Q_{\square} = \emptyset$. 
\item
A \emph{basis} consisting of a set $S$ of states of $\tree$ that are pairwise apart (i.e., form a clique).  Initially $S = \set{q^0}$. 
We maintain, for each nonbasis state $q \in Q_\tree \setminus S$, its \emph{candidate set}, that is, the set of basis states from which it is not apart:
\begin{equation*}
	C(q) = \set{p \in S \mid \lnot(q \apart p)}.
\end{equation*}
\item 
An integer $n$, which is the minimal number of states of our next hypothesis DFA. Initially $n = 1$, and we always have $\abs{S} \leq n$.
\end{enumerate} 

$\lsharpb$ starts with a membership query for the empty word, stores the outcome in the observation tree and $Q_{\square}$, and then repeatedly and nondeterministically applies the following four rules until a separating DFA has been found:

\paragraph{Promotion rule:}
Whenever we find a $q$ with with $C(q) = \emptyset$ (meaning that $q$ is apart from all basis states), we may add $q$ to basis $S$, and set $n$ to $\max(n, |S|)$.  

\paragraph{Extension rule:}
For a state $q \in S$ with access sequence $w = \access(q)$ and a word $v \in\Sigma^*$ up to length $n - \abs{S} + 1$, we may perform a membership query $wv$ (if we have not done so already).

\paragraph{Identification rule:}
For a state $q$ with $w = \access(q)$ that is a successor of a basis state, states $p, r \in C(q)$ and witness $v$ such that $v \vdash p \apart r$, we may perform a membership query $wv$ (if we have not done so already). If we receive a "$+$" or "$-$" answer, then $v$ also becomes a witness for $q \apart p$ or $q \apart r$. 

\paragraph{Validity rule:}
We may ask the SMT solver if there exists a complete DFA $\Hyp$ with at most $n$ states with a morphism $f : \Obs \to \Hyp$.
If the SMT solver answers ``no'' then we increment $n$.
Otherwise, if it answers ``yes'', we pose a validity query for the DFA $\Hyp$ provided by the solver.
If the teacher answers ``yes'' then the algorithm finishes and returns separating DFA $\Hyp$.
Otherwise, if the teacher provides a counterexample $w$, we pose membership queries for $w$ and all its prefixes (if we have not done so already).
This ensures that $\Hyp$ no longer will be proposed as an hypothesis by the solver.

\begin{example}
	\label{example run algorithm}
We illustrate the $\lsharpb$ algorithm on a simple example where $L_1$ is the language, accepted by the DFA of \Cref{Fig:Obs}(bottom), of words over $\set{a, b}$ with both an even number of $a$'s and an even number of $b$'s, and $L_2$ is the language consisting of all words over $\set{a, b}$ with odd length.
$\lsharpb$ initializes the observation tree with a single state $t_0$,
sets $Q_{\square}$ to $\emptyset$, $S$ to $\{ t_0 \}$ and $n$ to $1$. Via a membership query $\epsilon$, the algorithm learns that $t_0$ is accepting. 
Next the extension rule is applied twice and new states $t_1$ and $t_2$ with access sequences $a$ and $b$, respectively, are added to the tree.  Both states are rejecting.
Witness $\epsilon$ shows that states $t_0$ and $t_1$ are apart, and thus the promotion rule can be applied: state $t_1$ is added to the basis and $n$ is set to $2$.
The extension rule is applied twice to explore the new basis state $t_1$, and two new states $t_3$ and $t_4$ with access sequences $aa$ and $ab$, respectively, are added to the tree. 
Since $aa \in L_1$, state $t_3$ is accepting, and since $ab \not\in L_1$ and $ab \not\in L_2$, membership query $ab$ has a ``don't care'' answer, and state $t_4$ is added to $Q_{\square}$.
At this point the validity rule is applied. As we don't know where the outgoing transition of $t_1$ should go, there are two possible hypothesis models for $n=2$.  Suppose that the SMT solver returns the wrong one, $\Hyp_1$, shown at the top right of \Cref{Fig:ExampleLsharpSquare}.
The teacher will now return a counterexample, say $b b$. Via membership query $b b$\footnote{Actually, it is not necessary to query $b b$, since the learner may assume that the result is the opposite of the output of $b b$ in the hypothesis},
the learner finds out that the new state $t_5$ is accepting.
This time the SMT solver returns the correct hypothesis $\Hyp_2$, and a validity query confirms that this is a separating DFA.
Note that $\Hyp_2$ is smaller than the minimal DFA for $L_1$ of \Cref{Fig:Obs}(bottom).
The final PTA is shown at the left of \Cref{Fig:ExampleLsharpSquare}.
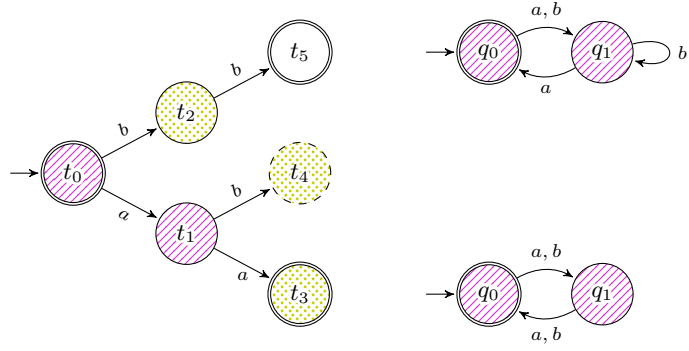
\begin{figure}[h] 
	\begin{center}
		\begin{tikzpicture}[->,>=stealth',shorten >=1pt,auto,node distance=1.5cm,main node/.style={circle,draw,font=\sffamily\large\bfseries},
			]
			\def\yoffset{8mm}
			\node[initial,accepting,state,basis] (0) {\treeNodeLabel{$t_0$}};
			\node[state,frontier] (2) [right of=0,yshift=\yoffset] {\treeNodeLabel{$t_2$}};
			\node[state, basis] (1) [right of=0,yshift=-\yoffset] {\treeNodeLabel{$t_1$}};
			\node[state,accepting,frontier] (3) [right of=1,yshift=-\yoffset] {\treeNodeLabel{$t_3$}};
			\node[state,dashed,frontier] (4) [right of=1,yshift=\yoffset] {\treeNodeLabel{$t_4$}};
			\node[state,accepting] (5) [right of=2,yshift=\yoffset] {\treeNodeLabel{$t_5$}};
			
			\node[initial,accepting,state,basis] [right of=5,xshift=1cm,basis](q0) {\treeNodeLabel{$q_0$}};
			\node[state,basis] (q1) [right of=q0] {\treeNodeLabel{$q_1$}};
			\node[initial,accepting,state,basis] (q2) [right of=3,xshift=1cm] {\treeNodeLabel{$q_0$}};
			\node[state,basis] (q3) [right of=q2] {\treeNodeLabel{$q_1$}};
			
			\path[every node/.style={font=\sffamily\scriptsize}]
			(0) edge node[sloped,above] {$b$} (2)
			(0) edge node[sloped,below] {$a$} (1)
			(1) edge node[below] {$a$} (3)
			(1) edge node[sloped,above] {$b$} (4)
			(2) edge node[sloped,above] {$b$} (5)
			
			(q0) edge [bend left, above] node {$a, b$} (q1)
			(q1) edge [bend left, below] node {$a$} (q0)
			(q1) edge [loop right] node {$b$} (q1)
			(q2) edge [bend left, above] node {$a, b$} (q3)
			(q3) edge [bend left, below] node {$a, b$} (q2);
		\end{tikzpicture}
		\caption{Final PTA (left), hypothesis $\Hyp_1$ (top right) and hypothesis $\Hyp_2$ (bottom right).}
		\label{Fig:ExampleLsharpSquare}
	\end{center}
\end{figure}
\end{example}

\paragraph{Remark.}
Unlike \lsharp, the $\lsharpb$ algorithm does not require that the basis is connected (prefixed closed).
This assumption is useful in \lsharp because construction of a hypothesis becomes trivial: just merge frontier states with nonconflicting basis states.  In \lsharp, without this assumption, some basis states may become unreachable in the hypothesis when we merge a preceding frontier state with another basis state. Thus \lsharp's hypothesis construction method no longer works when the basis is not prefix closed.  In $\lsharpb$ the SMT solver takes care of hypothesis construction and dropping prefix closure allows us to make progress with learning even when we fail to identify frontier states due to don't care answers.

\paragraph{Heuristics.}
$\lsharpb$ is a nondeterministic algorithm and its rules can be applied in an arbitrary order.
An implementation of $\lsharpb$, however, will typically use heuristics to select the next rule that will be applied, and this influences performance.
Thus, for instance, our implementation applies the four rules with the following priority order: promotion $>$ extension $>$ identification $>$ validity.
In particular, the validity rule (with its expensive call to an SMT solver) is used only when none of the other rules can be applied.
In applications where $n - |S|$ becomes big, the number of applications of the extension rule may grow exponentially, and a different heuristics may be required.
Our implementation just picks some arbitrary access sequence $w$ and suffix $v$ when it applies the extension rule.
Clever heuristics might help to select those sequences $w$ and $v$ for which the likelihood that membership query $w v$ produces an informative ($+$ or $-$) answer is maximal.
The main goal of processing a counterexample after a validity query $\Hyp$ is ensuring that $\Hyp$ is no longer a valid hypothesis. To this end, querying just $w$ would suffice, but we act under the assumption that prefixes of $w$ are more likely to not give a ``don't care'' answer.
Experimental evaluation shows that our approach works well for the benchmarks that we have tested. Clearly, the optimal learning strategy depends on the nature of the system under learning.

\subsection{SMT Encoding}
\label{smtencoding}
In order to decide whether there exists a complete DFA $\Hyp$ with at most $n$ states with a morphism $f : \Obs \to \Hyp$, we use the SMT solver Z3~\cite{z3}, which has support for the theory of uninterpreted functions.
We directly encode the definition of a morphism as given in \cref{def functional simulation}. This encoding is largely based on the encoding given in~\cite{smt}, but has been adapted to our setting and extended with redundant clauses. We introduce the following functions and encodings:
\begin{itemize}
        \item The states of $\Hyp$ are numbered $1$ to $n = \abs{Q_\Hyp}$.
        \item The states of $\tree$ are numbered $1$ to $m = \abs{Q_\tree}$; we use $I_q$ to denote the number associated with state $q \in Q_\tree$.
        \item We ensure that $I_q \leq \abs{S} (\leq n)$ for each $q \in S$
        \item The transition relation is encoded as a function $\delta_\Hyp : [1..n] \times \Sigma \to [1..n]$.
        \item The final state function is encoded by a function $F_\Hyp : [1..n] \to \{ 1, 0 \}$, with $1$ representing $+$, and $0$ representing $-$.
        \item The morphism is encoded by a function $f : [1..m] \to [1..n]$.
\end{itemize}

We introduce the following constraints: \begin{gather}
    \bigwedge_{q \in Q_\tree}\bigwedge_{\sigma \in \Sigma}\delta_\tree(q,\sigma)\darrow \implies f(I_{\delta_\tree(q,\sigma)}) = \delta_\Hyp(f(I_{q}), \sigma) \\
    \bigwedge_{q \in Q_\tree} F_\tree (q)\darrow \implies F_\tree(q) = F_\Hyp(f(I_{q})) \\
    \bigwedge_{q \in S}f(I_q) = I_q \\
    \bigwedge_{q \in Q_\tree \setminus S} \br{\br{\bigvee_{\abs{S} < k \leq n} f(I_q) = k} \lor \br{\bigvee_{p \in C(q)} f(I_{q}) = I_{p}}}
\end{gather}


Constraints (1) and (2) encode the definition of a morphism. Constraints (3) and (4) are redundant, but significantly speed up the solving time by specifying that two states that are apart in $\tree$ can never be mapped to the same state in $\Hyp$. Notice that the premises of the conditions in constraints (1) and (2) are known before calling the solver. Thus, to further speed up the solving time, we can simply include the conclusion itself as a constraint, only if the premise holds.

When the solver returns \texttt{SAT}, it is easy to extract a hypothesis from $\delta_\Hyp$ and $F_\Hyp$. In the case that the solver returns \texttt{UNSAT}, we know that no hypothesis exists.

\begin{proposition}
\label{smtcorrect}
    When an SMT solver with constraints (1) to (4) returns \texttt{SAT}, then there exists a complete DFA $\Hyp$ with at most $n$ states with a morphism $f : \Obs \to \Hyp$. Likewise, when the solver returns \texttt{UNSAT}, then there does not exist such a DFA.
\end{proposition}

\subsection{Correctness and Query Complexity of \texorpdfstring{\lsharps}{L\#-box}}
In this subsection, we will prove that the \lsharps algorithm is correct in the sense that it always terminates, and when it terminates, it returns a correct and minimal DFA. We will start by proving the minimality of the returned DFA. To this end, we will show that $n$ is only incremented when there is no complete DFA with $n$ states that the teacher would accept.

\begin{lemma}
\label{extensionlemma}
    Consider a learner that uses $\lsharpb$ to interact with a teacher that answers membership and validity queries for disjoint languages $L_1$ and $L_2$.
    Suppose that, at some point, the learner has constructed PTA $\Obs$, and suppose further that $\B$ is a complete separating DFA for $L_1$ and $L_2$.
    Then $\Obs \sqsubseteq \B$.
\end{lemma}

\begin{proof}
    Since $\B$ is a complete DFA, we may easily construct a unique mapping $f$ from $Q_\Obs$ to $Q_\B$ that satisfies the first and third condition of a morphism:
    $f(q^0_{\Obs}) = q^0_{\B}$ and
    $q \xrightarrow{a}_{\Obs} q'$ $\implies$ $f(q) \xrightarrow{a}_{\B} f(q')$.
    It remains to show that $f$ satisfies the second condition of a morphism.
    Suppose that, for some state $q \in Q_{\Obs}$, $F_{\Obs}(q)\defined$.
    Let $w = \access(q)$.  Then the learner has performed a membership query to which the teacher has replied, leading to value $F_{\Obs}(q)$ being defined.  There are two cases:
    (1) The teacher replied $+$.  This means that $w \in L_1$.
    Since $\B$ is a separating DFA for $L_1$ and $L_2$, $L_1 \subseteq L(\B)$.
    Hence $w \in L(\B)$.  This implies $F_{\Obs}(q) = F_{\B}(f(q)) = 1$, as required.
    (2) The teacher replied $-$. This means that $w \in L_2$.
    Since $\B$ is a separating DFA for $L_1$ and $L_2$, $L_2 \cap L(\B) = \emptyset$.
    Hence $w \not\in L(\B)$.  This implies $F_{\Obs}(q) = F_{\B}(f(q)) = 0$, as required.
\end{proof}

\begin{lemma}
\label{notsmallersmt}
    When the SMT solver returns \texttt{UNSAT}, then there is no complete DFA $\Hyp$ with $n$ states that would be accepted by the teacher.
\end{lemma}

\begin{proof}
    The proof follows from the second part of \cref{smtcorrect} and \cref{extensionlemma}.
\end{proof}

\begin{lemma}
\label{notsmallers}
    There does not exist a complete DFA $\Hyp$ with fewer than $\abs{S}$ states such that $\tree$ is a PTA for $\Hyp$.
\end{lemma}
\begin{proof}
    We know that all states in $S$ are pairwise apart. Thus, for any morphism $f$ from $\tree$ to $\Hyp$, we know from \cref{la: apartness refinement} that every state in $S$ must be mapped to a different state in $\Hyp$. Thus, $\Hyp$ needs at least $\abs{S}$ states.
\end{proof}

\begin{proposition}
    \label{correctn}
    When the \lsharps algorithm returns a DFA $\Hyp$, then $\Hyp$ is a minimal separating DFA for $L_1$ and $L_2$. In other words, there does not exist a (complete) DFA $\Hyp'$ with fewer states than $\Hyp$ that would be accepted by the teacher.
\end{proposition}
\begin{proof}
    From \cref{notsmallersmt} and \cref{notsmallers}, we know that whenever we increment $n$, then there does not exist a valid hypothesis with $n$ or fewer states. Thus, when the teacher accepts our hypothesis with $n$ states, then it would not accept any hypothesis with less than $n$ states.
\end{proof}

Next, we will show that \lsharps always terminates if $L_1$ is regular. We use the following bound on the number of regular languages that are accepted by a DFA with up to $N$ states.

\begin{lemma}[Ishigami \& Tani \cite{IshigamiT97}]
	\label{number of DFAs}
	Let $k$ be the number of input symbols in alphabet $\Sigma$, $N$ a natural number, and
	\begin{eqnarray*}
		\mathit{DFA}(N) & = & \{ L(\A) \mid \A \mbox{ is a DFA with at most } N \mbox{ states over alphabet } \Sigma \}.
	\end{eqnarray*}
	Then $| \mathit{DFA}(N) | \leq 2^{(k-1+o(1)) N \log N}$.
\end{lemma}

\begin{theorem}
	If $L_1$ is regular then the \lsharps algorithm terminates.
	The total number of membership queries is in $\bigO(kN 2^{k N \log N} )$, and the total number of validity queries is at most $2^{(k-1+o(1)) N \log N}$.
\end{theorem}
\begin{proof}
	If $L_1$ is regular then there is DFA $\A$ with $L(\A) = L_1$.  This DFA separates $L_1$ and $L_2$.  This implies that a minimal, complete separating DFA exists.
    Let $N$ be the number of states of a minimal, complete separating DFA for $L_1$ and $L_2$.
    We verify that each of the four rules can only be applied finitely often:
    \begin{itemize}
    	\item 
     	\emph{Promotion rule.}
     	By \Cref{extensionlemma} and \Cref{notsmallers}, the size of the basis is at most $N$, and thus the promotion rule can be applied at most $N - 1$ times.
     	\item 
     	\emph{Extension rule.}
    	The extension rule may explore, for each basis state, all outgoing paths of length at most $n - |S| + 1$.  Since $n$ is at most $N$ and $|S|$ is at least 1, the rule may explore all outgoing paths of length at most $N$ for each basis state.
    	Hence the extension rule may be applied at most
    	$N(1+ k + k^2 + \cdots + k^N) \in \bigO(N k^N)$ 
    	many times.
    	\item 
    	\emph{Validity rule.}
    	Since we increment $n$ whenever the SMT solver answers ``no'',
    	such a ``no'' answer may occur at most $N-1$ times during a run
    	of the algorithm.
    	In case of a ``yes'' answer, the learner poses a validity query.
        When it sends a hypothesis $\Hyp$ to the teacher and receives a counterexample $w$, then any $\Hyp'$ with $L(\Hyp') = L(\Hyp)$ is no longer a valid hypothesis after adding $w$ to $\tree$. From the first part of \cref{smtcorrect} we know that the SMT solver will only give a valid hypothesis. Thus,
        the total number of validity queries (and applications of the validity rule) is
        in $\bigO(2^{k N \log N})$, by \Cref{number of DFAs}.
    	\item
    	\emph{Identification rule.}
    	During each application of the identification rule, the learner performs a single membership query.
    	If this membership query does not return $\square$, then the candidate set of some successor of a basis state decreases in size. Since this can happen at most $\bigO(k N^2)$ times, from some moment on all membership queries performed during application of the identification rule will return $\square$.
    	This means that no new witnesses will be generated and therefore at some point we will have have explored all witnesses for all successors of all basis states, and no more applications of the identification rule are possible.
    	We may bound the number of applications of the identification rule as follows:
    	each identification query runs a witness starting from a frontier state.
    	The number of frontier states is bounded by $k N$.
    	The number of witnesses is bounded by the total number of queries resulting in a $+$ or $-$ answer, which in turn is bounded by the sum of
    	(a) the number of extension queries, which is in $\bigO(N k^N)$,
    	(b) the number of validity queries, which is in $\bigO(2^{k N \log N})$, and
    	(c) the number of identification queries resulting in a $+$ or $-$ answer, which is in $\bigO(k N^2)$.    	
    	Thus the number of identification queries is bounded by
    	$\bigO(kN 2^{k N \log N} )$.
\end{itemize}
The overall bound of $\bigO(kN 2^{k N \log N} )$ on the number of membership queries is directly implied by the derived bounds for the number of extension and identification queries.
\end{proof}

\section{Optimisations}
\label{optimisations}

\subsection{Promotion Based on Incompatibility}
\label{sec:promotion based on incompatibility}

The $\lsharpb$ algorithm maintains a basis consisting of states that are pairwise apart.  By \Cref{la: apartness refinement}, we know that if two states are apart, then they are incompatible, and thus correspond to different states of the minimal separating DFA that we are constructing.
The converse implication does not hold: incompatible states are not necessarily apart \cite{RW23}.

\begin{example}
	Consider the 3DFA $\C$ of \cref{Fig:Incompatible but not apart}. We claim that states $p$ and $q$ are incompatible but not apart.
	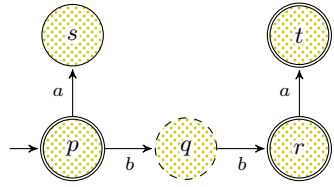
\begin{figure}[h] 
		\begin{center}
			\begin{tikzpicture}[->,>=stealth',shorten >=1pt,auto,node distance=1.5cm,main node/.style={circle,draw,font=\sffamily\large\bfseries},
				]
				\node[initial,accepting,state,frontier] (0) {\treeNodeLabel{$p$}};
				\node[state, frontier] (1) [above of=0] {\treeNodeLabel{$s$}};
				\node[state, frontier,dashed] (2) [right of=0] {\treeNodeLabel{$q$}};
				\node[state, accepting, frontier] (3) [right of=2] {\treeNodeLabel{$r$}};
				\node[state, accepting, frontier] (4) [above of=3] {\treeNodeLabel{$t$}};
				
				\path[every node/.style={font=\sffamily\scriptsize}]
				(0) edge node {$a$} (1)
				(0) edge node[below] {$b$} (2)
				(2) edge node[below] {$b$} (3)
				(3) edge node {$a$} (4)
				;
			\end{tikzpicture}
			\caption{States $p$ and $q$ are incompatible but not apart.}
			\label{Fig:Incompatible but not apart}
		\end{center}
	\end{figure}
	Clearly, states $p$ and $q$ are not apart: the only sequences that are enabled by both $p$ and $q$ are $\epsilon$ and $b$, but these do not lead to conflicting pairs of states.
	\ifshowappendix Appendix~\ref{sec:incompatibility} \else The full version \cite{LSV26report} \fi shows how RPNI state merging \cite{rpni} can be used to prove incompatibility of $p$ and $q$.
	Informally, the argument goes as follows.
	Assume $p$ and $q$ are compatible. 
	Then there is a morphism from $\C$ that maps $p$ and $q$ to the same state of some DFA $\B$.
	But then there also is a morphism to $\B$ from the 3NFA of \cref{Fig:Incompatible but not apart2} (left), obtained by merging $p$ and $q$.
	Because state $p$ now has two outgoing $b$ transitions and $\B$ is deterministic, there is also a morphism to $\B$ from the 3NFA of \cref{Fig:Incompatible but not apart2} (right), obtained by merging $p$ and $r$.
	This leads to a contradiction, as $p$ now has two outgoing $a$-transitions, one leading to an accepting state and one leading to a rejecting state: no morphism can map both $s$ and $t$ to the same state of $\B$.

	\begin{figure}[h] 
		\begin{center}
			\begin{tikzpicture}[->,>=stealth',shorten >=1pt,auto,node distance=1.5cm,main node/.style={circle,draw,font=\sffamily\large\bfseries},
				]
				\node[initial,accepting,state,frontier] (0) {\treeNodeLabel{$p$}};
				\node[state, frontier] (1) [above of=0] {\treeNodeLabel{$s$}};
				\node[state, accepting, frontier] (3) [right of=0] {\treeNodeLabel{$r$}};
				\node[state, accepting, frontier] (4) [above of=3] {\treeNodeLabel{$t$}};
				
				\path[every node/.style={font=\sffamily\scriptsize}]
				(0) edge node {$a$} (1)
				edge[loop below]  node[right] {$b$} (0)
				(0) edge node[below] {$b$} (3)
				(3) edge node {$a$} (4)
				;
			\end{tikzpicture}
			\hspace{1cm}
			\begin{tikzpicture}[->,>=stealth',shorten >=1pt,auto,node distance=1.5cm,main node/.style={circle,draw,font=\sffamily\large\bfseries},
				]
				\def\yoffset{8mm}
				\node[initial,accepting,state,frontier] (0) {\treeNodeLabel{$p$}};
				\node[state, frontier] (1) [above of=0] {\treeNodeLabel{$s$}};
				\node[state, accepting, frontier] (4) [right of=1] {\treeNodeLabel{$t$}};
				
				\path[every node/.style={font=\sffamily\scriptsize}]
				(0) edge node {$a$} (1)
				edge[loop below]  node[right] {$b$} (0)
				(0) edge node {$a$} (4)
				;
			\end{tikzpicture}
			\caption{Merging states $p$ and $q$ (left), and then merging $p$ and $r$ (right).}
			\label{Fig:Incompatible but not apart2}
		\end{center}
	\end{figure}
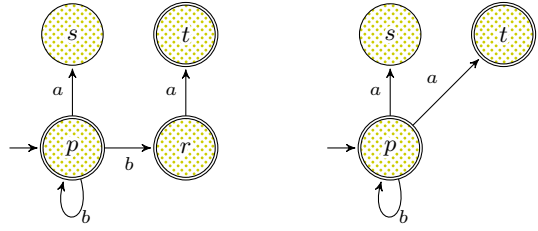
\end{example}

The example suggests an optimisation of $\lsharpb$ with a modified promotion rule, in which a state may be added to the basis when it is incompatible with all current basis states. In this way we maintain the crucial invariant that all states in the basis correspond to different states of the hidden DFA $\B$.
A key advantage of apartness over incompatibility is that witnesses for the apartness of basis states may subsequently be used to identify frontier states via the identification rule.
Therefore, whenever we add a state to the basis that is not apart from some basis state, it may pay off to perform some additional queries to establish apartness of incompatible basis states. This requires that the replies for these queries is either ``$+$'' or ``$-$'', which is not guaranteed in our setting.
In the example of \cref{Fig:Incompatible but not apart}, a ``$+$'' or ``$-$'' reply to query $b a$ establishes apartness of $p$ and $q$.
For details we refer to
\ifshowappendix
Appendix~\ref{sec:incompatibility}
\else
the full version of this article \cite{LSV26report}.
\fi

\subsection{Replacing Basis States}
During a run of the algorithm, it may occur that we add a state to the basis that is disconnected from the other basis states. Such a disconnected basis state can potentially be quite deep in the observation tree. However, it is often the case that nodes that are not as deep in the observation tree are more likely to be informative compared to nodes that are deeper in the tree. Furthermore, in practice, performing queries for words with fewer symbols can be faster than when using more symbols. Thus, it is desirable to choose a basis such that its states are as close as possible to the root of the observation tree.
To achieve this, we allow replacing deep basis states with closer basis states: whenever $q \in Q \setminus S$ with $C(q) = \set{p}$ and $\abs{\access(q)} < \abs{\access(p)}$, we may replace $p$ with $q$ in $S$.

\section{Evaluation}
\label{evaluation}
To evaluate the performance of the \lsharps algorithm, we implemented it in the AALpy~\cite{aalpygit} library for Python, and ran it against two groups of benchmarks, which we describe in the next subsections. 
%
The Oliveira and Silva experiments were performed on an AMD Ryzen 7 7800X3D processor with 32 GB of memory. The bug description experiments were performed on an AMD Epyc 7642 processor with 2 GB of memory.

\subsection{Oliveira and Silva Benchmarks}
We use the benchmarks of Oliveira and Silva~\cite{benchmarks}, which were also used to evaluate \lstars. These benchmarks are constructed as follows: For each size from 4 states to 23 states, 19 random DFAs are generated over the alphabet $\{ 0, 1 \}$. For each DFA, there are 5 sets of accepted and rejected words, each of these sets consisting of 20 random strings of length 30 and all their prefixes, along with the corresponding output of the DFA. This produces 95 benchmarks for each DFA size, for a total of 1900 benchmarks.

\lsharps was able to complete 1805 out of the 1900 benchmarks within 200 seconds, while \lstars was able to complete 1073 benchmarks. A table of the number of completed benchmarks is shown in \cref{fig:nocompleted}. For contrast, Heule and Verwer~\cite{passive} were able to complete all benchmarks of sizes 4 to 21 within 200 seconds. However, their algorithm requires knowing all observations from the start, which in practice is an easier problem to solve.

A direct comparison of the running time, number of membership queries, and number of validity queries for all completed benchmarks is shown in \cref{vs_running_time,vs_queries,vs_validity}. These graphs show the 25th, 50th, 75th and 100th percentile for each benchmark. Since the 0th percentile is very close to 0 each time, it is not included. The running time for \lsharps is dominated by the time spent by the SMT solver, taking over 95\% of the total for most benchmarks. We see that \lsharps shows significant improvement in the running time, number of membership queries, and number of validity queries.

\begin{figure}
\begin{center}
    \begin{tabular}{c|c|c|c|c|c|c|c|c|c|c}
        Benchmark & 4 & 5 & 6 & 7 & 8 & 9 & 10 & 11 & 12 & 13 \\
        \hline
        \lstars & 95 & 95 & 94 & 95 & 94 & 94 & 91 & 81 & 65 & 65 \\
        \lsharps & 95 & 95 & 95 & 95 & 95 & 95 & 95 & 95 & 95 & 95
    \end{tabular}
    
    \vspace{1ex}
    
    \begin{tabular}{c|c|c|c|c|c|c|c|c|c|c}
        Benchmark & 14 & 15 & 16 & 17 & 18 & 19 & 20 & 21 & 22 & 23 \\
        \hline
        \lstars & 48 & 34 & 30 & 23 & 14 & 22 & 4 & 14 & 11 & 4 \\
        \lsharps & 95 & 95 & 94 & 94 & 94 & 85 & 89 & 81 & 74 & 54
    \end{tabular}
\end{center}
\caption{Number of completed benchmarks for \lstars and \lsharps}
\label{fig:nocompleted}
\end{figure}
\begin{figure}
    \centering
    \includegraphics[width=\linewidth,alt={This plot shows, for each benchmark size in the range from $4$ to $23$, the running times in seconds on an exponential scale ranging from $10^{-2}$ to $2 \cdot 10^2$.  For each DFA size, the graph shows the 25th, 50th, 75th and 100th percentile for \lsharps and \lstars.  For all DFA sizes, \lsharps is faster than \lstars with up to two orders of magnitude.}]{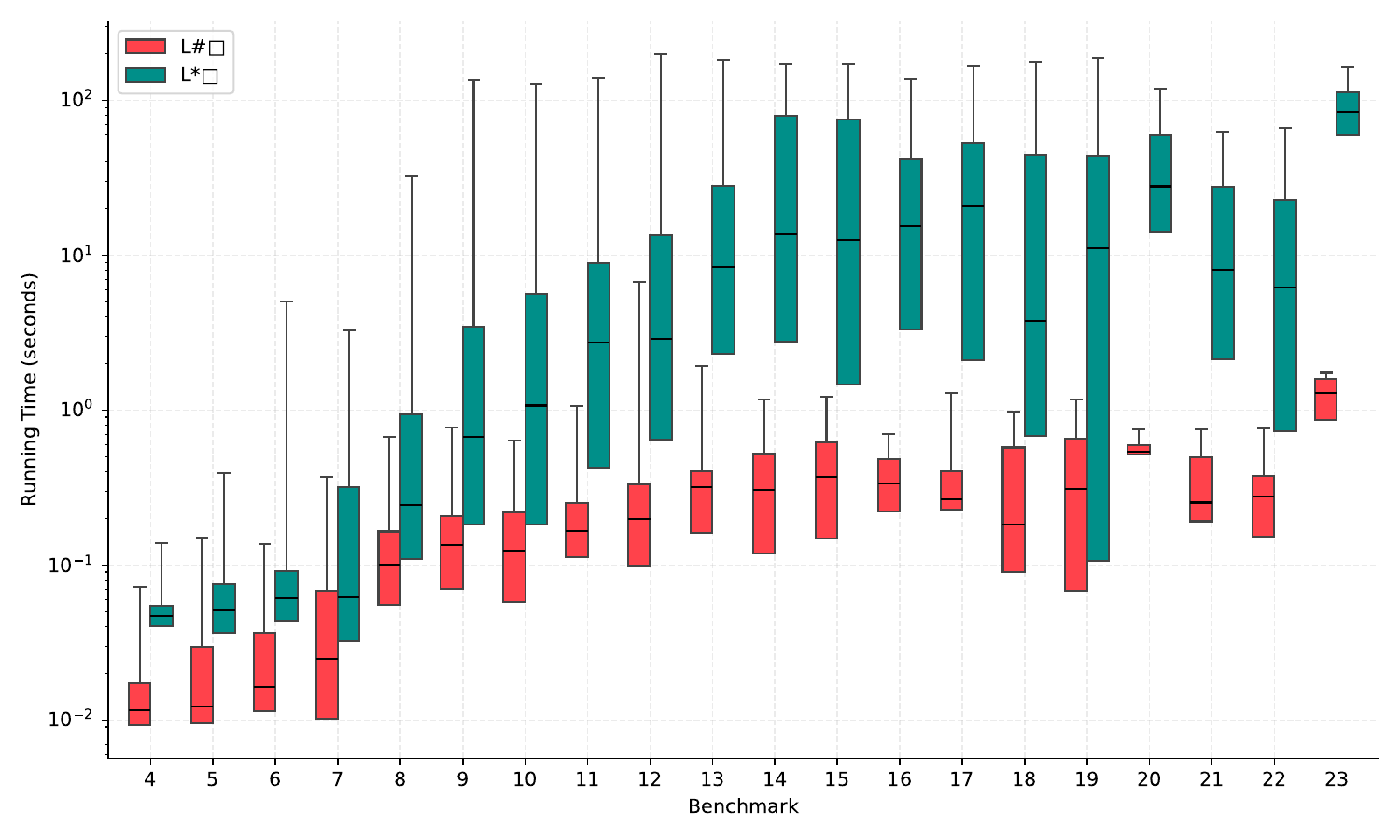}
    \caption{Comparison of running time of \lsharps and \lstars}
    \label{vs_running_time}
\end{figure}
\begin{figure}
    \centering
    \includegraphics[width=\linewidth,alt={This plot shows, for each benchmark size in the range from $4$ to $23$, the number of membership queries required to learn it on an exponential scale ranging roughly  $10^1$ to $2 \cdot 10^4$.  For each DFA size, the graph shows the 25th, 50th, 75th and 100th percentile for \lsharps and \lstars.  With the exception of minimal size 4, \lsharps is faster than \lstars.  From size 10 onwards, the difference is at least one order of magnitude.}]{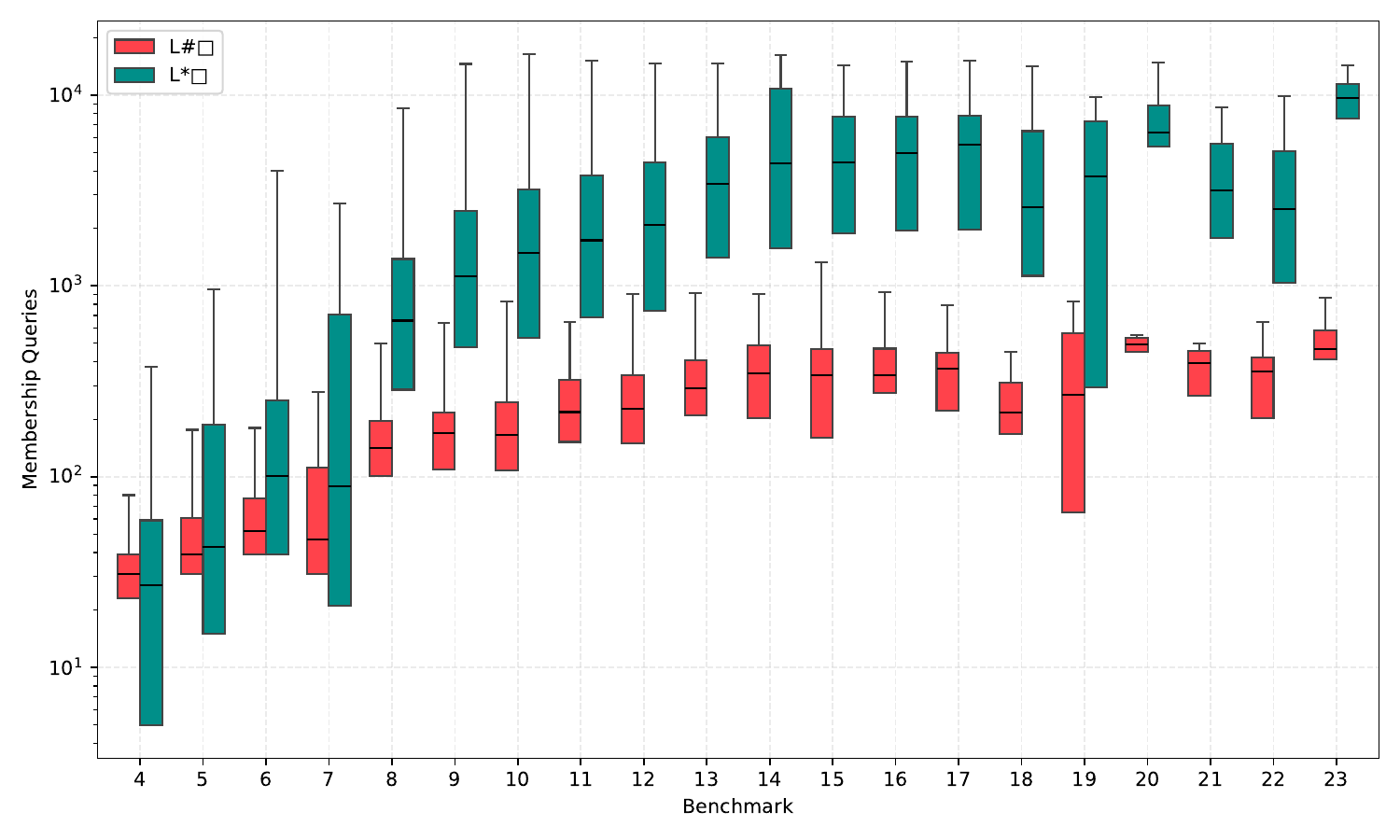}
    \caption{Comparison of membership queries of \lsharps and \lstars}
    \label{vs_queries}
\end{figure}
\begin{figure}
    \centering
    \includegraphics[width=\linewidth,alt={This plot shows, for each benchmark size in the range from $4$ to $23$, the number of validity queries required to learn it, which ranges from $1$ to $25$.  For each DFA size, the graph shows the 25th, 50th, 75th and 100th percentile for \lsharps and \lstars.  For all DFA sizes, \lsharps requires fewer queries than \lstars.  The number of required validity queries varies a lot for different benchmarks.}]{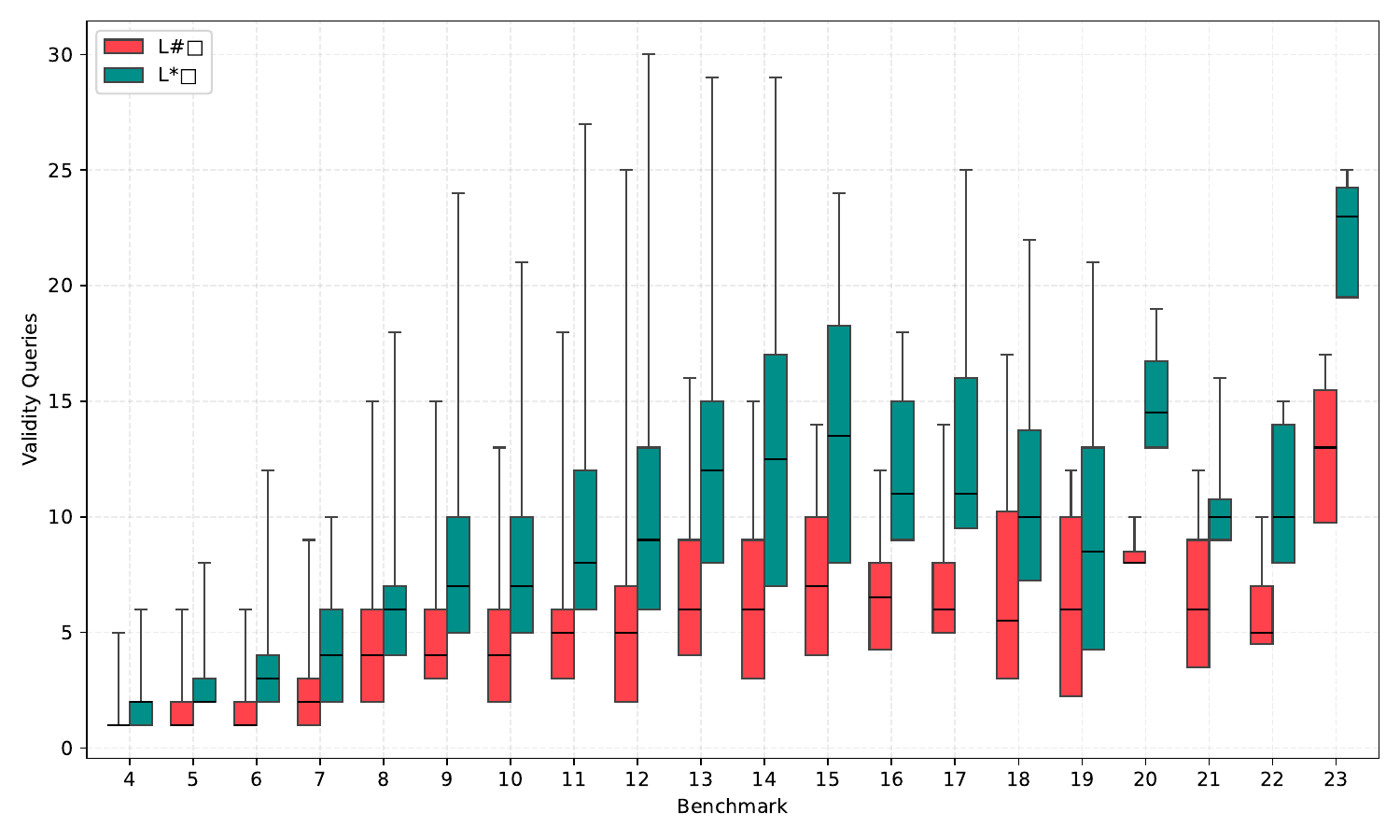}
    \caption{Comparison of validity queries of \lsharps and \lstars}
    \label{vs_validity}
\end{figure}

\subsection{Learning Bug Description DFAs}
We use the setup from~\cite{YaacovWAH25}, which is based on industrial benchmarks of control software components of ASML's TWINSCAN lithography machines from the RERS challenge 2019~\cite{rers}. These models are given in obfuscated Java code, each containing a bug. Our goal is to learn a DFA that accepts any valid input that leads to this bug, and rejects any valid input that does not lead to this bug.

To show the application of such a teacher, we consider a setting with a real online store system~\cite{store}. This store has actions for starting a session, logging in and out, adding to cart and removing from cart, and checking out. This store contains a bug, causing it to crash whenever an item is removed from the cart after adding it twice. Giving a direct description of this bug, where all inputs leading to a crash are accepted and all other inputs are rejected, leads to a DFA with 20 states. This DFA has to account for many behaviours that do not directly correspond to the bug, such as handling logins and logouts. Instead, using an incomplete teacher, we learn the DFA shown in \cref{fig:smallbug}, which exactly described the bug from earlier.

\begin{figure}[h] 
	\begin{center}
		\begin{tikzpicture}[shorten >=1pt, node distance=3cm, on grid, auto]
          \node[state, initial] (q0) {$q_0$};
          \node[state] (q1) [right=of q0] {$q_1$};
          \node[state] (q2) [right=of q1] {$q_2$};
          \node[state, accepting] (q3) [right=of q2] {$q_3$};
        
          \path[->]
            (q0) edge node {$\text{Add}(x)$} (q1)
            (q1) edge node {$\text{Add}(x)$} (q2)
            (q2) edge node {$\text{Remove}(x)$} (q3)
            (q0) edge[loop above] node {$\Sigma \setminus \text{Add}(x)$} ()
            (q1) edge[loop above] node {$\Sigma \setminus \text{Add}(x)$} ()
            (q2) edge[loop above] node {$\Sigma \setminus \text{Remove}(x)$} ()
            (q3) edge[loop above] node {$\Sigma$} ();
        \end{tikzpicture}
		\caption{Concise bug description of online store using an incomplete teacher}
        \label{fig:smallbug}
	\end{center}
\end{figure}
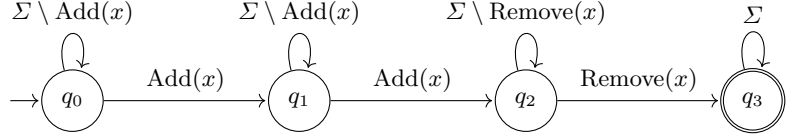

Yaacov et al.~\cite{YaacovWAH25} use a modified version of \lstar based on \cite{LeuckerN12,ChenFCTW09} to first learn a 3DFA, and then uses the RPNI algorithm~\cite{rpni} to construct a small DFA. With \lsharps, we are able to directly learn a minimal DFA, finding very similar bug description models more efficiently. A table with the results is shown in \cref{fig:lstar3dfa}. These results are taken as the mean of 10 runs. Besides the performance benefit, using \lsharps or other solver based algorithms has another benefit, that is, we can add additional constraints to the solver which might lead to more interpretable DFAs. For example, we can ask the solver to maximise the number of self-transitions.

Sometimes, the model learned by \lsharps is larger than the model learned by 3DFA-\lstar, even though \lsharps theoretically learns minimal models. This can be explained by the fact that we do not have access to a perfect validity oracle in this setting and instead, following \cite{YaacovWAH25}, we used a randomized version of the W-method \cite{Ch78,Vas73} implemented in AALpy. This can lead to the teacher sometimes not finding a counterexample even if it does exist, and cause a smaller model to be accepted.

\begin{figure}
\centering
\begin{tabular}{c|cr|cr||c|cr|cr}
     & 3DFA-\lstar & & \lsharps & 
     &  & 3DFA-\lstar & & \lsharps & \\
    Benchmark & size & time (s) & size & time (s)
    & Benchmark & size & time (s) & size & time (s) \\
    \hline
    m24 & 4 & 128 & 2 & 21
    & m22 & 5 & 271 & 3 & 41 \\
    m45 & 4 & 95 & 4 & 21
    & m27 & 4 & 81 & 3 & 38 \\
    m54 & 5 & 1098 & 5 & 209
    & m41 & 4 & 182 & 4 & 71 \\
    m55 & 7 & 839 & 3 & 32
    & m106 & 4 & 219 & 3 & 17 \\
    m76 & 4 & 668 & 3 & 34
    & m131 & 2 & 97 & 2 & 71 \\
    m95 & 4 & 73 & 3 & 19
    & m132 & 3 & 14605 & 3 & 909 \\
    m135 & 4 & 181 & 2 & 15
    & m167 & 3 & 112 & 3 & 22 \\
    m158 & 4 & 31 & 4 & 17
    & m173 & 4 & 1460 & 3 & 122 \\
    m159 & 4 & 41 & 4 & 24
    & m182 & 4 & 1373 & 2 & 44 \\
    m164 & 2 & 30 & 2 & 13
    & m189 & 3 & 71 & 2 & 28 \\
    m172 & 4 & 90 & 2 & 10
    & m196 & 9 & 18824 & 4 & 75 \\
    m181 & 4 & 17584 & 2 & 337
    & m199 & 4 & 109 & 3 & 24 \\
    m183 & 4 & 161 & 6 & 414
    &  &  &  &  \\
    m185 & 4 & 665 & 3 & 24
    &  &  &  &  \\
    m201 & 4 & 469 & 3 & 25
    &  &  &  &  \\
\end{tabular}
\caption{Comparison of \lstar for 3DFAs and \lsharps}
\label{fig:lstar3dfa}
\end{figure}

\subsection{Evaluation of Optimisations}
To evaluate the performance of the optimisations of \cref{optimisations} and the redundant clauses of the SMT solver, we performed an ablation-style study
for the benchmarks of Oliveira and Silva~\cite{benchmarks}.

The use of redundant clauses in the SMT encoding leads to a significant improvement in the running time\ifshowappendix, as shown in \cref{vs_no_red}\fi. This effect increases for larger benchmarks. For benchmarks of sizes 12 and above, the algorithm already becomes impractical without the redundant clauses.

Replacing basis states with states that have shorter access sequences leads to a small improvement, and reduces the number of input symbols in membership queries with $7\%$\ifshowappendix, as shown in \cref{vs_no_replace}\fi.

Finally, requiring compatibility of all pairs of basis states, rather than apartness, does not lead to a significant improvement. The running time for the version with compatibility is only $2\%$ faster on average\ifshowappendix, as shown in \cref{vs_comp}\fi. Other measurements also do not show any significant improvements.
\ifshowappendix
$\mbox{}$
\else

We refer to the report version of this article \cite{LSV26report} for additional details.
\fi

\section{Discussion}
\label{sec:discussion}

We introduced $\lsharpb$, a simple algorithm for learning a minimal separating DFA using membership and validity queries.
We see several possibilities for further improving $\lsharpb$, e.g., the SMT encoding can be further optimized using ideas from \cite{passive} and/or incremental SMT solving, and we may develop heuristics to predict for which membership queries we can expect a useful ($+$ or $-$) answer by learning an approximation of $L_1 \cup L_2$.

When Angluin \cite{Ang87} introduced the notion of \emph{Minimally Adequate Teacher (MAT)}, she argued that a teacher is only ``adequate'' when it is possible for a learner to learn a concept efficiently using only a polynomial number of queries.
Even when $L_1$ and $L_2$ are regular, we have no polynomial bound on the number of queries required by $\lsharpb$, and indeed it is an open problem whether teachers that answer membership+validity queries are adequate, in the sense that they allow a learner to learn a separating DFA with a polynomial number of queries.
So one might argue that the name {\sf iMAT} for the learning framework of \cite{lstars} was chosen prematurely.
Adequacy is obtained trivially when the teacher, besides membership queries, also answers equivalence queries for $L_1$ and $\overline{L_2}$.  A learner may then first learn DFAs for $L_1$ and $L_2$ using a state-of-the-art DFA learner, and then compute a separating automaton using an SMT solver.
The query complexity of such an approach is $\bigO(| \Sigma | (n_1 + n_2)^2 + (n_1 + n_2) \log m)$, where $n_1$ and $n_2$ are the number of states of the minimal DFAs for $L_1$ and $L_2$, respectively, and $m$ is the length of the longest counterexample.
As equivalence queries are incomparable with validity queries, one might argue that membership+equivalence queries provide a MAT framework for learning separating automata.
However, a naive algorithm that uses equivalence queries is impractical as it requires one to solve a big instance of an NP-hard problem.
Clearly, the trade off between query and time complexity of learning algorithms deserves further study.

Even in a restricted setting where no ``don't care'' answer is ever returned, the $\lsharpb$ and \lsharp algorithms behave differently. Whereas \lsharp poses extra queries (using binary search) to analyze counterexamples, $\lsharpb$ does not have or need counterexample processing. Also, unlike \lsharp, $\lsharpb$ does not require the basis to be prefix closed. Whereas \lsharp extends the basis with a single input at a time, $\lsharpb$ may use longer sequences.  It will be interesting to compare both algorithms in a setting without ``don't care'' answers, but we expect \lsharp to be far more efficient in terms of the required number of queries.

An obvious direction for future research is to extend our new algorithm to richer model classes such as Mealy machines, symbolic automata \cite{DAntoniV21}, register automata \cite{DierlFHJST24}, weighted automata \cite{BalleCLQ14}, and automata with timers \cite{BruyereGPSV25}.
The SMT encoding used by $\lsharpb$ for hypothesis construction only involves propositional logic and uninterpreted functions, but the
ability of SMT solvers to also handle real numbers, integers, and various data structures may effectively support learning of richer automata classes, even in settings without ``don't care'' answers.

\subsubsection*{Acknowledgements.}
We thank Sebastian Junges for a fruitful discussion on an earlier version of this work, and the anonymous reviewers for their valuable feedback.
This work is partially sponsored by NWO project OCENW.M.23.155, Evidence-Driven Black-Box Checking (EVI).

\subsubsection*{Disclosure of Interests.}
The authors have no competing interests to declare that are relevant to the content of this article.

\subsubsection*{Data-Availability Statement.}
We implemented the \lsharps algorithm in the AALpy library~\cite{aalpygit} for Python, using Z3 ~\cite{z3} as SMT solver.
The source code is available at \url{https://github.com/JLaumen/l-sharp-square-algorithm}.
The artifact available at \url{https://doi.org/10.5281/zenodo.19816690} makes it possible to reproduce all the experiments described in this article. The artifact consists of two Docker files that may be used to reproduce our results for the benchmarks of Oliveira and Silva~\cite{benchmarks}, and the benchmarks from~\cite{YaacovWAH25} for learning bug descriptions, respectively.

\bibliographystyle{splncs04}
\bibliography{references}

\ifshowappendix
\newpage
\appendix
\section{Figures for Ablation Study}
\label{sec:ablationfigures}

\begin{figure}
	\centering
	\includegraphics[width=\linewidth]{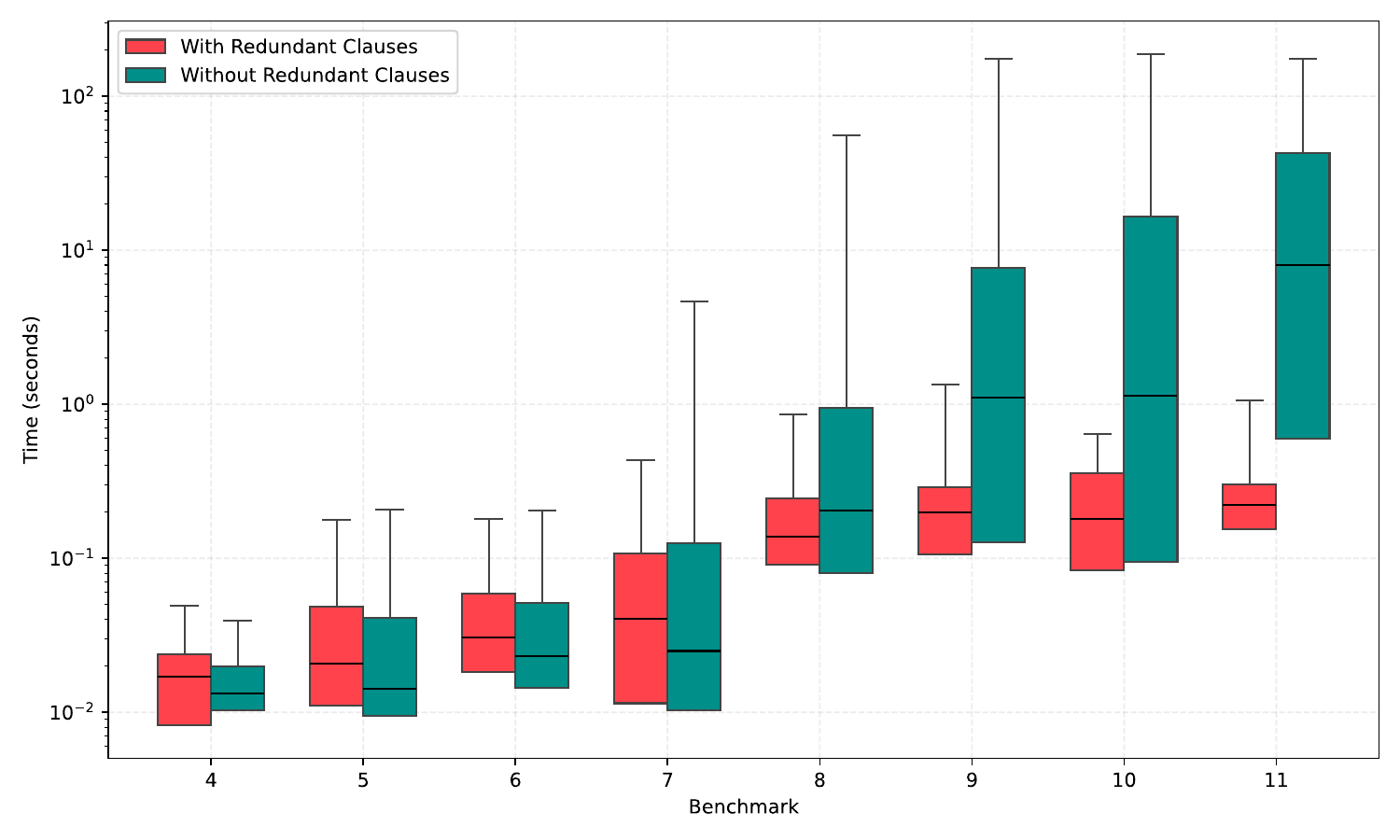}
	\caption{Comparison of running time of redundant clauses}
	\label{vs_no_red}
\end{figure}

\begin{figure}
	\centering
	\includegraphics[width=\linewidth]{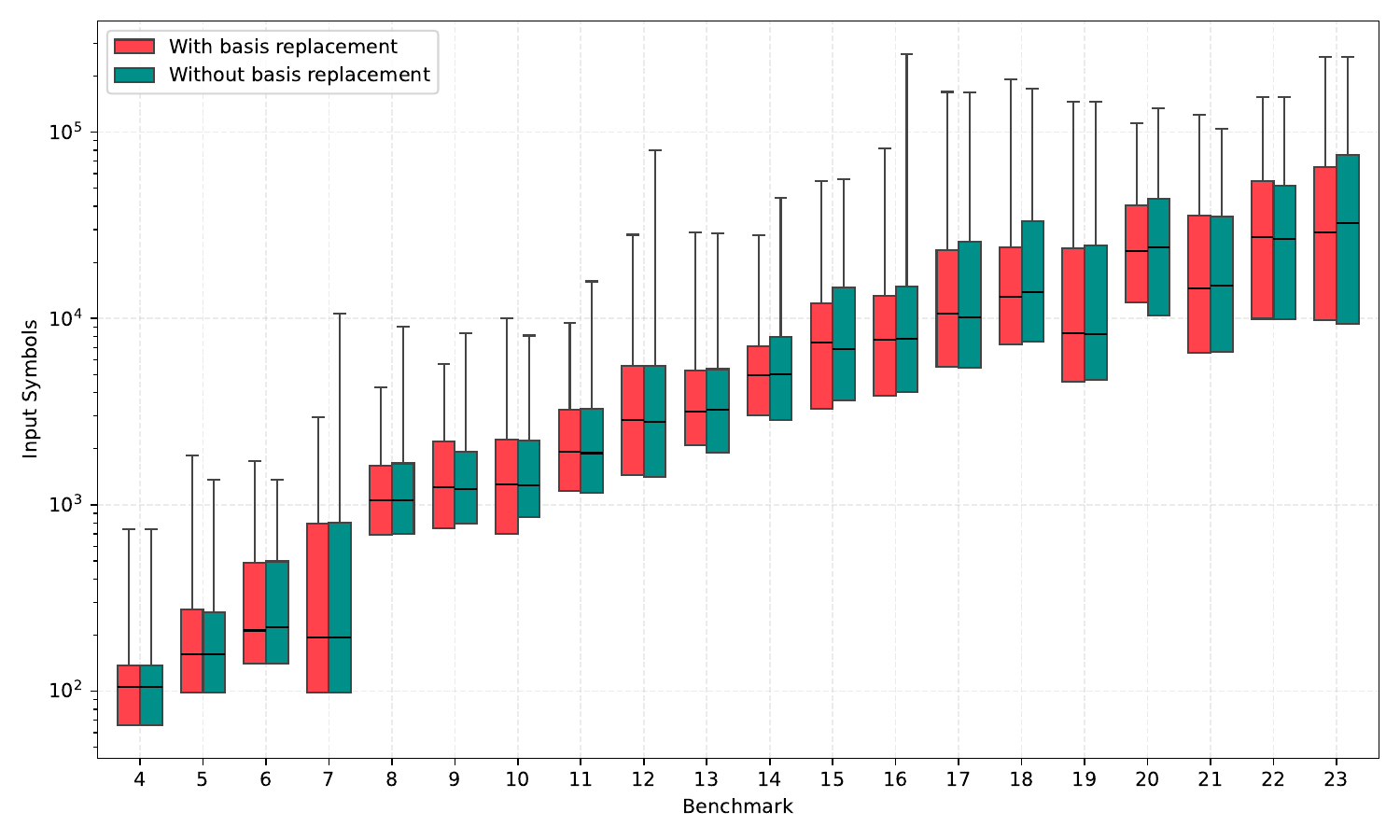}
	\caption{Comparison of input symbols of basis replacement}
	\label{vs_no_replace}
\end{figure}

\begin{figure}
	\centering
	\includegraphics[width=\linewidth]{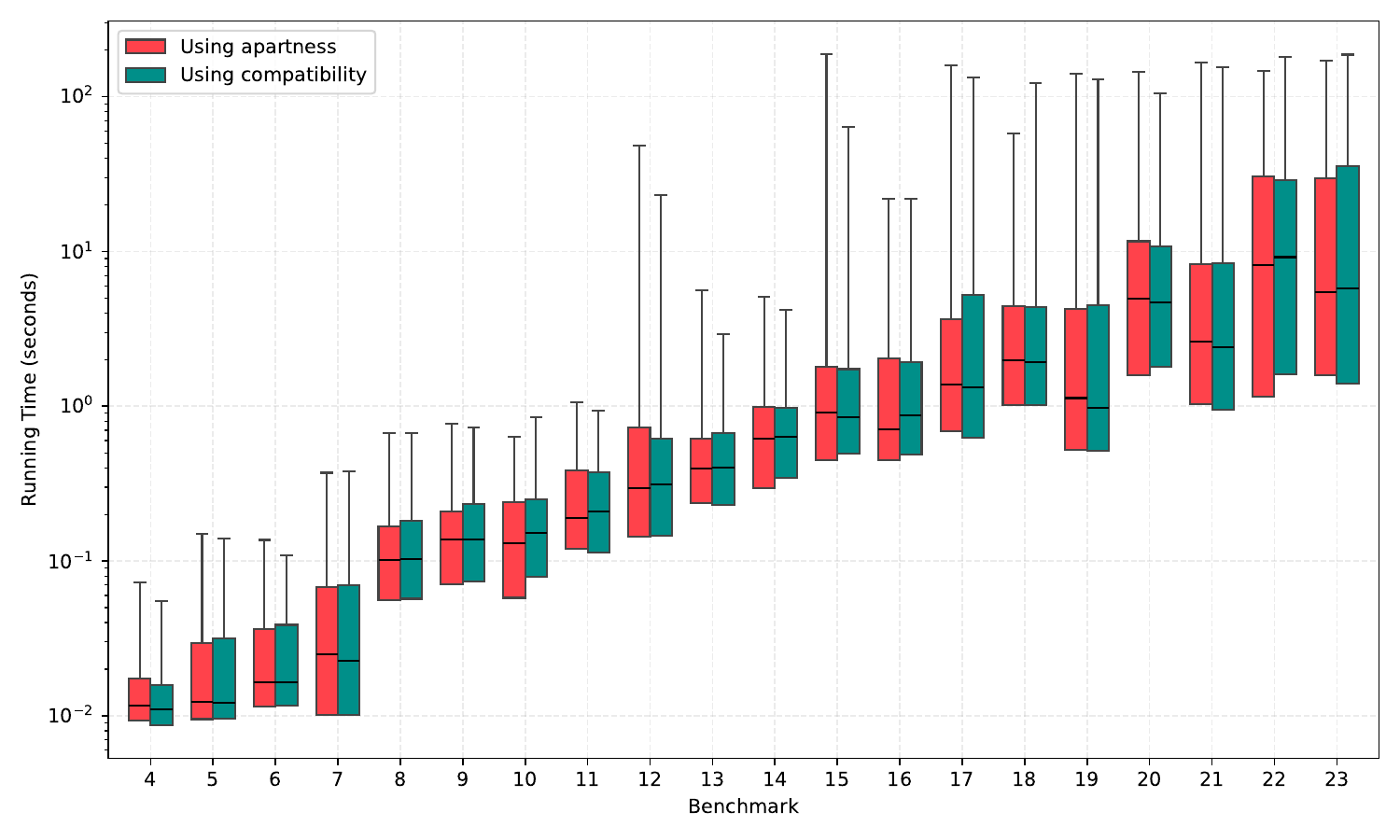}
	\caption{Comparison of running time of apartness and compatibility}
	\label{vs_comp}
\end{figure}

\newpage
\section{State Incompatibility and Apartness}
\label{sec:incompatibility}

For every 3NFA $\A$, an equivalent 3DFA can be constructed using the classical subset construction \cite{HU79}, but this may lead to an exponential blowup of the number of (reachable) states.
In this appendix, we consider an alternative, stronger notion of determinizability: the existence of a 3DFA $\B$ with $\A \sqsubseteq_f \B$, for some $f$.
W.l.o.g.\ we may assume that all states of $\B$ are in the image of $f$, which implies that the number of states of $\B$ is less than or equal to the number of states of $\A$.
Not every 3NFA is strongly determinizable and we prove how, using the RPNI state merging algorithm \cite{rpni}, we may efficiently decide strong determinizability.

As a corollary, we also obtain an algorithm for deciding incompatibility.
We then show how, based on a run of the RPNI algorithm that establishes incompatibility of states $p$ and $q$, we may compute a collection of membership queries that, when answered with a ``yes'' or ``no'''answer by the teacher, establish apartness of $p$ and $q$.

\begin{definition}
	An 3NFA $\A$ is called \emph{strongly determinizable} if there exists a 3DFA $\B$ such that $\A \sqsubseteq \B$.
\end{definition}

Note that, instead of requiring existence of a 3DFA, it is equivalent to require existence of a DFA in the above definition: (a)  Since any DFA is a 3DFA, existence of a DFA $\B$ with $\A \sqsubseteq \B$ trivially implies existence of a 3DFA $\B$ with $\A \sqsubseteq \B$. (b) Conversely, suppose that $f$ is a morphism from $\A$ to a 3DFA $\B$.  Let $\B'$ be the DFA obtained from $\B$ by making all states with unknown status accepting.
Note that for any state $q$ of $\B$ with unknown status, all states of $\A$ that are mapped to $q$ by $f$ also have unknown status.  Hence $f$ is also a morphism from $\A$ to $\B'$.

\begin{lemma}
	\label{lemma not mergeable implies not SD}
	Suppose $\A$ is an 3NFA with states $q$, $q'$ and $q''$ such that $q \xrightarrow{a} q'$ and $q \xrightarrow{a} q''$, for some $a \in \Sigma$, and $q'$ and $q''$ are conflicting.
	Then $\A$ is not strongly determinizable.
\end{lemma}
\begin{proof}
	By contradiction.  Assume that $\A$ is strongly determinizable.
	Then there exists a 3DFA $\B$ and a morphism $f$ from $\A$ to $\B$.
	Since $f$ is a morphism,
	$f(q) \xrightarrow{a}_{\B} f(q')$ and $f(q) \xrightarrow{a}_{\B} f(q'')$.
	Using that $\B$ is deterministic gives $f(q') = f(q'')$.
	Since $q'$ and $q''$ are conflicting, we conclude that
	$F_{\A}(q')\converges$, $F_{\A}(q'')\converges$, and $F_{\A}(q') \neq F_{\A}(q'')$.
	Using again that $f$ is a morphism, we derive
	\[
	F_{\B}(f(q')) = F_{\A}(q') \neq F_{\A}(q'') = F_{\B}(f(q''))
	\]
	Thus $f(q') \neq f(q'')$
	This contradicts our earlier observation that $f(q') = f(q'')$.
\end{proof}

The $\mathsf{merge}$ function takes as arguments a 3NFA $\A$ and two nonconflicting states $p$ and $q$, and returns a 3NFA in which these states are merged together.

\begin{definition}[State merging]
	\label{def state merging}
	Let $\A$ be an 3NFA with distinct, nonconflicting states $p$ and $q$.
	Then $\statemerge{\A}{p}{q}$ is the 3NFA $\B$ with:
	\begin{itemize}
		\item 
		$Q_{\B} = Q_{\A} \setminus \{ q \}$,
		\item 
		$q^0_{\B} = \begin{cases}
			p & \mbox{if } q = q^0_{\A}\\
			q^0_{\A} &\mbox{otherwise}
		\end{cases}$,
		\item 
		For all $r \in Q_{\B}$, $F_{\B}(r) = \begin{cases}
			F_{\A}(q) & \mbox{if } r=p \wedge F_{\A}(q)\converges\\
			F_{\A}(r) & \mbox{otherwise}
		\end{cases}$,
		\item 
		Let $\alpha \colon Q_{\A} \to Q_{\B}$ be the function that maps each state of $Q_{\A}$ to itself, except for $q$ which is mapped to $p$. Then $\rightarrow_{\B}$ is the smallest relation such that, for all $q, q'\in Q_{\A}$ and $a\in \Sigma$,
		$q \xrightarrow{a}_{\A} q'$ implies $\alpha(q) \xrightarrow{a}_{\B} \alpha(q')$. 
	\end{itemize}
\end{definition}

\begin{lemma}
	\label{state merging functional simulation}
	Let $\A$ be an 3NFA with distinct, nonconflicting states $p$ and $q$, and let
	$\alpha \colon Q_{\A} \to Q_{\B}$ be the function used in \Cref{def state merging}.
	Then $\A \sqsubseteq_{\alpha} \statemerge{\A}{p}{q}$.
\end{lemma}
\begin{proof}
	Let $\B = \statemerge{\A}{p}{q}$.
	We check that function $\alpha$, satisfies the conditions of a morphism (for all $r, r'\in Q_{A}$ and $a \in \Sigma$):
	\begin{enumerate}
		\item $\alpha(q^0_{\A}) = q^0_{\B}$.  There are two cases:
		\begin{itemize}
			\item 
			If $q = q^0_{\A}$ then $q^0_{\B} = p = \alpha(q) = \alpha(q^0_{\A})$.
			\item
			If $q \neq q^0_{\A}$ then  $q^0_{\B} = q^0_{\A} = \alpha(q^0_{\A})$.
		\end{itemize}
		\item $F_{\A}(r)\defined$ $\implies$ $F_{\A}(r) = F_{\B}(\alpha(r))$. There are seven cases:
		\begin{enumerate}
			\item 
			$r \neq p \wedge r  \neq q$.
			Then $F_{\B}(r) = F_{\A}(r)$, by definition of $F_{\B}$.
			Also, $F_{\B}(r) = F_{\B}(\alpha(r))$, by definition of $\alpha$. Hence $F_{\A}(r) = F_{\B}(\alpha(r))$.
			\item 
			$r=p \wedge F_{\A}(p) \diverges$.
			Trivial, since $F_{\A}(r) \diverges$.
			\item 
			$r=q \wedge F_{\A}(q) \diverges$.
			Trivial, since $F_{\A}(r) \diverges$.
			\item 
			$r=p \wedge F_{\A}(q) \diverges$.
			Then $F_{\B}(r) = F_{\B}(p) = F_{\A}(p)$, by definition of $F_{\B}$.  
			Also, $F_{\B}(r) = F_{\B}(\alpha(r))$, by definition of $\alpha$. Hence $F_{\A}(r) = F_{\B}(\alpha(r))$.
			\item 
			$r=q \wedge F_{\A}(p)\diverges \wedge F_{\A}(q)\converges$.
			Then  $F_{\B}(p) = F_{\A}(q)$, by definition of $F_{\B}$.
			Also, $F_{\B}(p) = F_{\B}(\alpha(q))$, by definition of $\alpha$. Hence $F_{\A}(r) = F_{\B}(\alpha(r))$.
			\item 
			$F_{\A}(p)\converges \wedge F_{\A}(q)\converges \wedge r = p$. Then $F_{\A}(p) = F_{\A}(q)$, since $p$ and $q$ are nonconflicting.
			Moreover $F_{\A}(q) = F_{\B}(p)$, by definition of $F_{\B}$.
			Also, $F_{\B}(p) = F_{\B}(\alpha(p))$, by definition of $\alpha$. By combining these observations, we obtain $F_{\A}(r) = F_{\B}(\alpha(r))$.
			\item 
			$F_{\A}(p)\converges \wedge F_{\A}(q)\converges \wedge r = q$. Then $F_{\A}(q) = F_{\B}(p)$, by definition of $F_{\B}$.
			Also, $F_{\B}(p) = F_{\B}(\alpha(q))$, by definition of $\alpha$. By combining these observations, we obtain $F_{\A}(r) = F_{\B}(\alpha(r))$.
		\end{enumerate}
		\item $r \xrightarrow{a}_{\A} r'$ $\implies$ $\alpha(r) \xrightarrow{a}_{\B} \alpha(r')$. This is immediate from the definition of the transition relation of $\statemerge{\A}{p}{q}$.
	\end{enumerate}
\end{proof}

\begin{lemma}
	\label{lemma identified by simulation implies mergeable}
	Suppose $\A$ and $\B$ are 3NFAs, $\A \sqsubseteq_f \B$, and $p, q \in Q_{\A}$ are distinct states with $f(p)=f(q)$. 
	Then $p$ and $q$ are nonconflicting.
\end{lemma}
\begin{proof}
	Assume $F_{\A}(p) \converges$ and $F_{\A}(q) \converges$.
	Then, since $f$ is a morphism,
	$F_{\A}(p) = F_{\B}(f(p))$ and $F_{\A}(q) = F_{\B}(f(q))$.
	Using $f(p) = f(q)$, we derive $F_{\A}(p) = F_{\B}(f(p)) = F_{\B}(f(q)) = F_{\A}(q)$.
	Hence states $p$ and $q$ are nonconflicting.
\end{proof}

\begin{lemma}
	\label{decomposition functional simulations}
	Suppose $\A$ and $\B$ are 3NFAs, $\A \sqsubseteq_f \B$, and $p, q \in Q_{\A}$ are distinct states with $f(p)=f(q)$. 
	Then $\statemerge{\A}{p}{q} \sqsubseteq_g \B$, where $g$ is the restriction of $f$ to $Q_{\A} \setminus \{ q \}$.
\end{lemma}
\begin{proof}
	Note that, by \cref{lemma identified by simulation implies mergeable}, states $p$ and $q$ are nonconflicting and so
	$\C = \statemerge{\A}{p}{q}$ is well-defined.
	Let $\alpha \colon Q_{\A} \to Q_{\C}$ be the function that maps each state of $Q_{\A}$ to itself, except for $q$ which is mapped to $p$. Then $f = g \circ \alpha$.
	We check that $g \colon Q_{\C} \to Q_{\B}$ satisfies the three conditions of a morphism (for all $r, r'\in Q_{\C}$ and $a \in \Sigma$):
	\begin{enumerate}
		\item $g(q^0_{\C}) = q^0_{\B}$. We consider two cases:
		\begin{enumerate}
			\item 
			$q = q^0_{\A}$. Then, by definition of $\C$, $q^0_{\C} = p$. We derive
			\[
			g(q^0_{\C}) = g(p) = f(p) = f(q) = f(q^0_{\A}) = q^0_{\B}.
			\]
			\item 
			$q \neq q^0_{\A}$. Then, by definition of $\C$, $q^0_{\C} = q^0_{\A}$. We derive
			\[
			g(q^0_{\C}) = g(q^0_{\A}) = f(q^0_{\A}) = q^0_{\B}.
			\]
		\end{enumerate}
		\item $F_{\C}(r)\defined$ $\implies$ $F_{\C}(r) = F_{\B}(g(r))$.
		Assume $F_{\C}(r)\defined$. We consider two cases:
		\begin{enumerate}
			\item 
			$r=p \wedge F_{\A}(q) \converges$.
			In this case, by definition of $\C$, $F_{\C}(r) = F_{\A}(q)$. We derive:
			\[
			F_{\C}(r) = F_{\A}(q) = F_{\B}(f(q)) = F_{\B}(g \circ \alpha(q)) = F_{\B} (g(p)) = F_{\B}(g(r)).
			\]
			\item 
			$\neg (r=p \wedge F_{\A}(q) \converges)$.
			In this case, by definition of $\C$, $F_{\C}(r) = F_{\A}(r)$. In particular, $F_{\A}(r) \converges$.
			We derive:
			\[
			F_{\C}(r) = F_{\A}(r) = F_{\B}(f(r)) = F_{\B}(g(r)).
			\]
		\end{enumerate}
		\item $r \xrightarrow{a}_{\C} r'$ $\implies$ $g(r) \xrightarrow{a}_{\B} g(r')$.
		Now suppose $r \xrightarrow{a}_{\C} r'$.
		Then, by definition of $\C$, there is a transition
		$s \xrightarrow{a}_{\A} s'$ such that $\alpha(s) = r$ and $\alpha(s') = r'$.
		Therefore, as $f$ is a morphism,
		$f(s) \xrightarrow{a}_{\B} f(s')$.
		Now observe that $f(s) = g \circ \alpha(s) = g(r)$ and
		$f(s') = g \circ \alpha(s') = g(r')$.
		Hence $g(r) \xrightarrow{a}_{\B} g(r')$, as required.
	\end{enumerate}
\end{proof}

\begin{lemma}
	\label{lemma reduce SD to merge}
	Suppose $\A$ is an 3NFA with states $q$, $q'$ and $q''$ such that $q \xrightarrow{a} q'$ and $q \xrightarrow{a} q''$, for some $a \in \Sigma$, and $q'$ and $q''$ are distinct and nonconflicting.
	Then $\A$ is strongly determinizable iff
	$\statemerge{\A}{q'}{q''}$ is strongly determinizable.
\end{lemma}
\begin{proof}
	$\mbox{}$
	\begin{itemize}
		\item 
		``$\Rightarrow$''
		Suppose that $\A$ is strongly determinizable.
		Then there exists a 3DFA $\B$ with a morphism $f$ from $\A$ to $\B$. 
		Observe that, as $f$ is a morphism,
		$f(q) \xrightarrow{a}_{\B} f(q')$ and $f(q) \xrightarrow{a}_{\B} f(q'')$.
		Since $\B$ is deterministic, this implies $f(q') = f(q'')$.
		This means we may use \Cref{decomposition functional simulations} to obtain $\statemerge{\A}{q'}{q''} \sqsubseteq \B$. Hence $\statemerge{\A}{q'}{q''}$ is strongly determinizable.
		\item 
		``$\Leftarrow$'' Suppose that $\statemerge{\A}{q'}{q''}$ is strongly determinizable.
		Then there exists a 3DFA $\B$ such that $\statemerge{\A}{q'}{q''} \sqsubseteq \B$.
		By \cref{state merging functional simulation},
		$\A \sqsubseteq \statemerge{\A}{q'}{q''}$.
		Now application of\cref{preorder} gives $\A \sqsubseteq \B$. Hence $\A$ is strongly determinizable.
	\end{itemize}
\end{proof}

\Cref{lemma not mergeable implies not SD} and \Cref{lemma reduce SD to merge} suggest a simple algorithm to decide strong determinizability of a 3NFA $\A$.
We say that $\A$ \emph{allows a merge} if
there are states $q, q', q'' \in Q_{\A}$ and an input symbol $a \in \Sigma$ such that $q \xrightarrow{a} q'$, $q \xrightarrow{a} q''$ and $q', q''$ are distinct and nonconflicting.
As long as $\A$ allows a merge, we may replace $\A$ by $\statemerge{\A}{q'}{q''}$ by \Cref{lemma reduce SD to merge}. This will terminate since 3NFA $\statemerge{\A}{q'}{q''}$ has one state less than 3NFA $\A$. 
If no more merge is possible, then we check whether the resulting 3NFA is deterministic.  
If it is then $\A$ is strongly determinizable. If it isn't then there are $q, q', q'' \in Q_{\A}$ and an input symbol $a \in \Sigma$ such that $q \xrightarrow{a} q'$, $q \xrightarrow{a} q''$ and $q', q''$ are conflicting, and
$\A$ is not strongly determinizable by \Cref{lemma not mergeable implies not SD}.
The pseudocode is shown in \Cref{alg:strong determinizability}.

\begin{algorithm}[h!]
	\begin{algorithmic}
		\Procedure{SD}{$\A$}
		\While{$\A$ allows a merge}
		\State find $q, q', q'' \in Q_{\A}$, $a \in \Sigma$ s.t.
		\State \hspace{2cm} $q \xrightarrow{a} q'$, $q \xrightarrow{a} q''$ and $q', q''$ distinct and nonconflicting
		\State $\A \gets \statemerge{\A}{q'}{q''}$
		\EndWhile
		\If{$\A$ is deterministic}
		\State \Return ``yes''
		\Else 
		\State \Return ``no''
		\EndIf
		\EndProcedure
	\end{algorithmic}
	\caption{Deciding whether a 3NFA $\A$ is strongly determinizable}
	\label{alg:strong determinizability}
\end{algorithm}

The next lemma implies that \cref{alg:strong determinizability} can be used to decide compatibility.

\begin{lemma}
	\label{la: compatible iff strongly determinizable}
	Let $p$ and $q$ be distinct states of a 3NFA $\A$.
	Then $p$ and $q$ are compatible iff they are nonconflicting and $\statemerge{\A}{p}{q}$ is strongly determinizable.
\end{lemma}
\begin{proof}
	We prove both implications:
	\begin{itemize}
		\item 
		``$\Rightarrow$'' Suppose states $p$ and $q$ are compatible.
		Then $\A \sqsubseteq_f \B$, for some 3DFA $\B$ and morphism $f$ with $f(p)=f(q)$.
		By \cref{lemma identified by simulation implies mergeable}, states $p$ and $q$ are nonconflicting.
		Next we apply \Cref{decomposition functional simulations} to conclude that $\statemerge{\A}{p}{q} \sqsubseteq \B$.
		Hence $\statemerge{\A}{p}{q}$ is strongly determinizable,
		as required.
		\item 
		``$\Leftarrow$'' Suppose $p$ and $q$ are nonconflicting and $\statemerge{\A}{p}{q}$ is strongly determinizable. Then $\statemerge{\A}{p}{q} \sqsubseteq_g \B$, for some 3DFA $\B$ and morphism $g$.
		By \Cref{state merging functional simulation}, $\A \sqsubseteq_f \statemerge{\A}{p}{q}$, for some morphism $f$ with $f(p)=f(q)$.
		Hence, by (the proof of) \Cref{preorder}, $\A \sqsubseteq_{g \circ f} \B$.
		Note that $g \circ f(p)= g \circ f(q)$.
		Hence states $p$ and $q$ are compatible.
	\end{itemize}
\end{proof}

If two states are apart and have incoming transitions with the same label, then trivially the source states of these transitions are also apart.

\begin{lemma}[Back propagation]
	\label{lemma back propagation}
	Let $\A$ be a 3NFA with states $p$, $p'$, $q$ and $q'$. Let $a \in \Sigma$.
	Then $p \xrightarrow{a} p'$ and $q \xrightarrow{a} q'$ and $p' \apart q'$ implies $p \apart q$.
\end{lemma}

The apartness relation satisfies a weaker version of
\emph{co-transitivity}, stating that if $w \vdash r\apart r'$ and from $q$ there is a $w$ path to a known state, then $q$ must be apart from $r$ or $r'$ (or both).
\begin{lemma}[Weak co-transitivity]
	\label{la: weak co-transitivity}
	Let $\A$ be a 3NFA. Then, for all states $p$, $q$, $r$ and $r'$, and for all $w \in \Sigma^*$,
	\[
	w \vdash p \apart q ~\wedge~ r \xRightarrow{w} r' ~\wedge~ F_{\A}(r') \converges  ~~\Longrightarrow~~
	p \apart r ~\vee~ q \apart r
	\]
\end{lemma}


In the $\lsharpb$ algorithm, rather than requiring that states in the basis  are pairwise
apart, we may just require them to be pairwise incompatible. 
However, a key advantage of apartness over incompatibility is that witnesses for the apartness of basis states may subsequently be used to identify frontier states.
Therefore, whenever we add a state to the basis that is not apart from some of the basis states, it may pay off to perform some additional queries and try to establish apartness of incompatible basis states
(queries that may or may not be answered by the incomplete teacher).

In order to check whether two distinct states $p$ and $q$ of a prefix tree acceptor $\Obs$ are incompatible, we first check if they are conflicting.
If they are, then $p$ and $q$ are apart (and thus incompatible) with witness $\epsilon$.
Otherwise, by \Cref{la: compatible iff strongly determinizable}, we apply \Cref{alg:strong determinizability} to determine whether
$\statemerge{\A}{p}{q}$ is strongly determinizable.
Basically, there are three cases to consider.
In the first case, the subtrees of $p$ and $q$ are disjoint and share no states. The merges of \Cref{alg:strong determinizability} will then squeeze the subtrees of $p$ and $q$ together, except when it encounters two $a$-transitions leading to conflicting states.
The result of running \Cref{alg:strong determinizability} will be a tree again, and $p$ and $q$ are compatible iff this tree is deterministic.  Moreover, $p$ and $q$ are incompatible iff they are apart.
In the second case, $q$ is contained in the subtree of $p$, and in the third case $p$ is contained in the subtree of $q$.  As the second and third case are symmetric, we will only cover the second case.
Suppose the unique path from $p$ to $q$ in $\Obs$ consists of $c$ transition (note $c>0$ as $p$ and $q$ are distinct).
Then $\statemerge{\A}{p}{q}$ contains exactly one loop, consisting of $c-1$ transitions from $p$ to the parent $q'$ of $q$, followed by a transition from $q'$ to $p$.
Subsequent merges by \Cref{alg:strong determinizability} may roll up one path from $q$ around the loop, like a ball of wool, but it will not create any new loops.

\begin{example}
	As an example, \Cref{Fig:Wool} shows an observation tree $\Obs$ and the NFA that results when we want to determine whether states $t_0$ and $t_2$ are compatible, that is, we run \Cref{alg:strong determinizability} on $\statemerge{\Obs}{t_0}{t_2}$.
	Morphism $f$, denoted via coloring of states, maps the states of $\Obs$ to states of $\B$.
	Since the resulting NFA $\B$ is nondeterministic, we conclude that states $t_0$ and $t_2$ are incompatible. However, the two states are not apart.
	Note that $f$ maps states $t_0$, $t_2$, $t_4$ and $t_6$ of $\Obs$ to state $t_0$ of $\B$,
	and states $t_1$, $t_3$ and $t_5$ of $\Obs$ to state $t_1$ of $\B$.
	Also note that $eb \vdash t_0 \apart t_6$ and $d \vdash t_1 \apart t_5$ in $\Obs$.
	Via back propagation, $d \vdash t_1 \apart t_5$ implies $a d \vdash t_0 \apart t_4$.
	Using weak co-transitivity and back propagation, we may try to establish apartness of $t_0$ and $t_2$:
	\begin{itemize}
		\item 
		We run witness $eb$ from state $t_2$, that is, we pose a membership query $a b e b$.
		If the teacher answers with ``no'' then $eb \vdash t_0 \apart t_2$ (and we are done). If the teacher answers with ``yes'' then we infer $eb \vdash t_2 \apart t_6$.
		\item 
		If $eb \vdash t_2 \apart t_6$, then we run witness $eb$ from $t_4$, that is, we pose membership query $ababeb$.  If the teacher answers ``yes'' then $t_4 \apart t_6$, which implies $t_0 \apart t_2$ by back propagation (and we are done).
		If the teacher answers ``no'' then $t_2 \apart t_4$, which implies $t_0 \apart t_2$ by back propagation (and we are done)
		\item
		If  the teacher did not answer the preceding queries, we may run witness $d$ from $t_3$, that is, pose query $abad$.  If the query is answered with ``yes'' then $t_1 \apart t_3$, which implies $t_0 \apart t_2$ (and we are done).  If the answer is ``no'' then $t_3 \apart t_5$, which implies $t_0 \apart t_3$ by back propagation (and we are done).
	\end{itemize}
	
	\begin{figure}[h] 
		\begin{center}
			\begin{tikzpicture}[->,>=stealth',shorten >=1pt,auto,node distance=1.5cm,main node/.style={circle,draw,font=\sffamily\large\bfseries},
				]
				\def\yoffset{8mm}
				\node[initial,accepting,state,basis] (0) {\treeNodeLabel{$t_0$}};
				\node[state, frontier] (1) [right of=0] {\treeNodeLabel{$t_1$}};
				\node[state, basis,dashed] (2) [right of=1] {\treeNodeLabel{$t_2$}};
				\node[state, frontier,dashed] (3) [right of=2] {\treeNodeLabel{$t_3$}};
				\node[state, basis,dashed] (4) [right of=3] {\treeNodeLabel{$t_4$}};
				\node[state, frontier,dashed] (5) [right of=4] {\treeNodeLabel{$t_{5}$}};
				\node[state, basis,dashed] (6) [right of=5] {\treeNodeLabel{$t_{6}$}};
				\node[state, q0class] (7) [above of=0] {\treeNodeLabel{$t_{7}$}};
				\node[state,q1class,accepting] (8) [above of=7] {\treeNodeLabel{$t_8$}};
				\node[state, q0class] (11) [above of=6] {\treeNodeLabel{$t_{11}$}};
				\node[state, q2class] (12) [above of=11] {\treeNodeLabel{$t_{12}$}};
				\node[state, accepting,q4class] (10) [above of=5] {\treeNodeLabel{$t_{10}$}};
				\node[state, q3class] (9) [above of=1] {\treeNodeLabel{$t_9$}};
				
				\node[initial,accepting,state,basis] [below of=3, yshift=-2*\yoffset](q0) {\treeNodeLabel{$t_0$}};
				\node[state,frontier] (q1) [right of=q0] {\treeNodeLabel{$t_1$}};
				\node[state,q0class] (q2) [below of=q0] {\treeNodeLabel{$t_7$}};
				\node[state,q1class,accepting] (q3) [below left of=q2] {\treeNodeLabel{$t_8$}};
				\node[state,q2class] (q4) [below right of=q2] {\treeNodeLabel{$t_{12}$}};
				\node[state,q3class] (q5) [above right of=q1] {\treeNodeLabel{$t_9$}};
				\node[state,q4class,accepting] (q6) [below right of=q1] {\treeNodeLabel{$t_{10}$}};
				
				\node[anchor=base] (f) at ($ (3.base) !.5! (q1.base) $) {\begin{tikzcd}
						{} \arrow{r}{f}
						&[8mm] {}
				\end{tikzcd}};
				
				\path[every node/.style={font=\sffamily\scriptsize}]
				(0) edge node {$e$} (7)
				(0) edge node[below] {$a$} (1)
				(1) edge node[below] {$b$} (2)
				(1) edge node {$d$} (9)
				(2) edge node [below] {$a$} (3)
				(5) edge node {$d$} (10)
				(6) edge node {$e$} (11)
				(3) edge node [below] {$b$} (4)
				(7) edge node {$b$} (8)
				(4) edge node [below] {$a$} (5)
				(11) edge node {$b$} (12)
				(5) edge node [below] {$b$} (6)
				(q0) edge [bend left, above] node {$a$} (q1)
				(q0) edge  node {$e$} (q2)
				(q1) edge [sloped, above] node {$d$} (q5)
				(q1) edge [sloped, below] node {$d$} (q6)
				(q1) edge [bend left, below] node {$b$} (q0)
				(q2) edge  node {$b$} (q3)
				(q2) edge  node {$b$} (q4)
				;
			\end{tikzpicture}
			\caption{A prefix tree acceptor $\Obs$ (top) for a NFA $\B$ (bottom).}
			\label{Fig:Wool}
		\end{center}
	\end{figure}
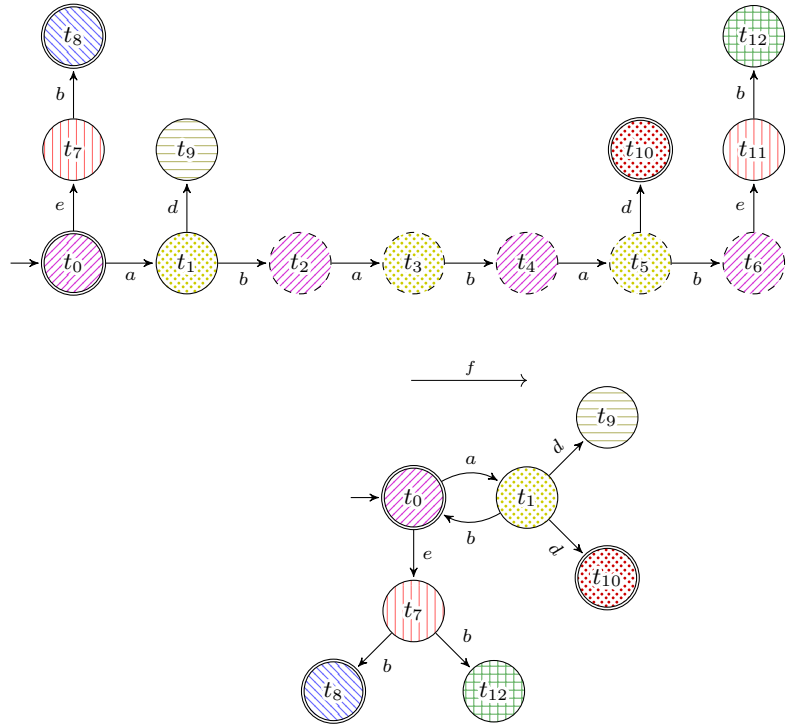
\end{example}	

In general, provided the teacher answers queries, we may always establish apartness of incomplete basis states:
(a) let $c$ be the number of transitions from $p$ to $q$.
(a) There must be two states $r_0$ and $r_k$ of $\Obs$ that are apart and mapped to the same state $s$ on the cycle in $\B$ (otherwise $p$ and $q$ would be compatible).
(b) Let $w$ be a witness for $r_0 \apart r_k$ and let $r_1 ,\ldots, r_{k-1}$ the states in between $r_0$ and $r_k$ that are also mapped to $s$. 
Then, for all $0 \leq j < k$, there is a path consisting of $c$ transitions from $r_j$ to $r_{j+1}$. If we run query $w$ from each state $r_1 ,\ldots, r_{k-1}$ and all queries are answered, there must be an index $0 \leq j < k$ such that $r_j \apart r_{j+1}$.
(Actually, we may find index $j$ using binary search).
(c) Now $p \apart q$ follows by back propagation.
\fi

\end{document}